\documentclass[12pt]{article}

\usepackage[utf8]{inputenc}
\usepackage{makeidx}
\usepackage{hyperref}
\usepackage{soul}
\makeindex

\usepackage{color}
\setstcolor{red}
\setul{}{0.25ex}

\title{The equivalence of local-realistic and no-signalling theories}
\author{Paul Raymond-Robichaud 
\footnote{
 Gilles Brassard was an author  of the \href{https://arxiv.org/abs/1710.01380v1}{first version}  of this paper. Indeed, he improved tremendously the presentation of the ideas   and participated in innumerable discussions with the author. Nevertheless, he retrospectively felt it was unfair to have been an author, especially the first author, given that many people falsely attributed the work and the ideas to him, while he was the voice of this work rather than its creator.  He proposed to withdraw from the article  in order to let his student shine and ensure proper attribution to the ideas.  I assured him that he could remain as a second author, but
  this was not an option for him because, as he had already explained to me
  years earlier, alphabetical order for authors' names is customary in
  his fields, and he has followed it strictly throughout his career.
 Thus,  he could have been an author and he continues to endorse with enthusiasm the ideas presented here. I  dedicate this paper to him and his invaluable friendship, since this work would  not have been possible without him.  }
 \\ 
{\normalsize ISI Foundation, Turin, Italy}\\ 
{\normalsize paul.r.robichaud@gmail.com}
}

\usepackage{verbatim}
\usepackage{lipsum}

\usepackage{amsthm}
\usepackage{times}
\usepackage{amsmath,amssymb}
\usepackage[english]{babel}
\usepackage[utf8]{inputenc}

\usepackage{tikz}
\usepackage{url}
\usetikzlibrary{decorations.pathreplacing}

\usetikzlibrary{matrix}
\usepackage[matrix,frame,arrow]{xypic}

\theoremstyle{definition}
\newtheorem{definition}{Definition}[section]

\newtheorem{lemma}{Lemma}[section]

\newtheorem{theorem}{Theorem}[section]

\newtheorem{postulate}{Postulate}[section]

\newtheorem{axiom}{Axiom}[section]

\newcommand{\isdef}{\stackrel{\mathrm{def}}{=}}

\newcommand{\prob}{\mathrm{Prob}}

\begin{document}

\maketitle
 
\begin{abstract}
We provide a framework that describe all local-realistic theories and all \mbox{no-signal}\-ling theories. We show that every local-realistic theory is a no-signalling theory. We also show that every no-signalling theory with invertible dynamics has a local-realistic model.
This applies in particular to unitary quantum theory.
\end{abstract}

\pagebreak

\section*{Acknowledgements}

My deepest gratitude goes to Gilles Brassard, who could
have been an author of this paper and whose role was explained in the title page footnote.

I am grateful to David Deutsch  for countless stimulating discussions and for providing numerous suggestions for improvement to a draft of this article.

I acknowledge stimulating discussions with Stefan Wolf, Renato Renner, Sandu Popescu, Marcin Paw{\l}owski,  Dominic Mayers,  Chiara Marletto, St\'ephane Durand, Giulio Chiribella,   Jeff Bub, Michel Boyer,  Charles Bennett,  and Charles Alexandre B\'edard.

This work has been supported in part by the Natural Sciences and Engineering Research Council of Canada, the Fonds de recherche du Qu\'ebec -- Nature et technologies and from Intesa Sanpaolo Innovation Center. The funders had no role in study design, data collection, and analysis, decision to publish, or preparation of the manuscript.

\section{Introduction}
This article presents original formal axioms defining  local-realistic theories 
and no-signalling theories.  These axioms attempt to be  the most general possible and are required to prove formally any  result that apply to  \emph{all} local-realistic theories and \emph{all} no-signalling theories.

The axioms  are formulated without any reference to interpretations of the theory.
However, the interpretation with which a given  no-signalling theory is presented may include additional statements  beyond its strict formalism,  which once properly formalized might be  inconsistent  with the axioms of local-realism or with a particular local-realistic model of it.   Thus, while the axioms of local-realism  do not make assumptions regarding the interpretation of a given theory, their supposition   nevertheless rules out interpretations of it.

This is in sharp contrast with the local-hidden variable formulation of local-realism introduced by Bell \cite{Bell64}.  This formulation   implicitely add interpretative assumptions  beyond the notion of local-realism \cite{ParLives, DH}.  Whenever a theory has a local-realistic model that cannot be described by local-hidden variables, this merely prove that  the  model is incompatible with these extraneous assumptions.
The simplest example of such a theory is the non-local box, introduced by Popescu and Rorhlich, which has a local-realistic interpretation and yet cannot be described by local-hidden variables  \cite{FreeWill, ParLives, PR, POSTER}. The most interesting example is certainly quantum theory, since it is the current theory of Nature, and David Deutsch and Patrick Hayden have proven that it has a local-realistic model, despite all its seemingly non-local features like entanglement \cite{DH, QMlocal}.

Once the notion of  local-realistic theories and no-signalling theories is properly axiomatised, we shall see that every local-realistic theory is trivially also a no-signalling theory. 
The reason is that the defining property of  a no-signalling theory is that  no action on a system can have any  \emph{observable} effect on a separated system,  while a key property of a local-realistic theory is that no  action   on a system can have any   effect \emph{whatsover} on a  separated system.

Conversely, we shall then prove   that under very minor postulates  \emph{any} given invertible-dynamics no-signalling theory   has a local-realistic model.
The proof of this fact emerged from the construction of a rigorous local-realistic model for quantum theory and from a formalization of the concept of local-realism  \emph{independent} of quantum theory \cite{QMlocal}.    
Once this was done, it soon appeared that  the principles defining local-realistic theories and the core ideas involved in  the creation of a local-realistic model for quantum theory could be used to  create a local-realistic model for   an arbitrary no-signalling theory with invertible dynamics.

Thus, this paper shares many ideas of the previous companion paper \cite{QMlocal}, including the same informal principles of local-realism, the same framework for local-realistic theories, except for the omission of Axiom \ref{ax:faithnoume}, which was not essential and  for the projection axioms (Axioms \ref{ax:noumeproj} and \ref{tracepheno}), which were previous presented in term of traces. Also, the theorems and proofs presented here emerged from previous theorems and proofs. However, in the current paper, quantum theory is viewed as a theory among many, and our sole interest in it  is  to show that it satisfies the formal requirements of a no-signalling theory with invertible dynamics, and thus that the construction presented here applies to it.
   Moreover,  we shall discuss  the philosophical aspects of local-realism and the axioms of local-realistic theories in much greater detail and without specific reference to quantum theory and finally various theorems that were not needed or proved in the previous paper are proved here.

All the key  ideas in the current paper  were presented in a chapter of the thesis of the author, performed under the supervision of Gilles Brassard \cite{PRRthesis}.  However, contrarily to this previous presentation, here local realistic theories and no-signalling theories are described through explicit axioms, and the hypotheses used to proved the main result are described by explicit postulates.  Equally the material in the appendices is new.   Finally,  for reasons of concision  the material in the thesis on how the axioms of local-realism could be extended to deal with aspects concerning infinitely many systems has been omitted.

\section{Conceptual foundations of local-realism}

Before stating the axioms that characterize a local-realistic theory, it would be useful to investigate  the concept of realism and local-realism in order to motivate some mathematical properties that will be explained rigorously later.  
But first,  what is realism?

Realism~is the idea that  their exists  a world  outside of our immediate subjective experience and that the state of this world determines the outcome of all observations. 

\subsection{Appearance vs reality}\label{sc:concepts}

This outside world can be called the \emph{real world},  
the \emph{external} world, the \emph{objective} world, or in Kant's terminology: the \emph{noumenal} world~\cite{Kant-Brit,Kant}.
It~describes the world as it is
rather than the world as it can observed, or known through sensory experience.  

According to realism, our subjective experience, our perceptions, our sense-data, are determined by the state of the external world. The portion of the real world that is observable or perceptible is called the \emph{perceptive world}, or alternatively the \emph{observable} world, or in Kant's terminology: the \emph{phenomenal} world. 

Kant's terminology shall be followed for the remainder of the article.  However this should neither be taken as endorsement of Kant's metaphysics nor as claim that the terms are employed in the same way.

To be perceptible does not mean to be perceived directly.  If~we scan a molecule with an atomic force microscope,
the properties thus observed are perceptible even though we are not observing the molecule directly with our naked eyes. The only limit to what kind of measurement device may be used to determine what is perceptible are the laws of Nature, not merely the currently available technology.  Also, to be perceptible does not imply to be perceived right now by some observer. 
For~instance, the far side of the Moon existed as part of the phenomenal world even before we had the technology that allowed us to observe~it.
Thus, we include as part of the perceptible world, not what is perceived now, but rather potential perceptions.

What is the relation between the noumenal world and the phenomenal world? What is perceptible must follow a process parallel to what exists.  As the noumenal world evolves, so does the phenomenal world. Any property that exists in the phenomenal world arises from a property in the noumenal world.   

We can represent the relation between the noumenal and the phenomenal worlds with the following diagram.

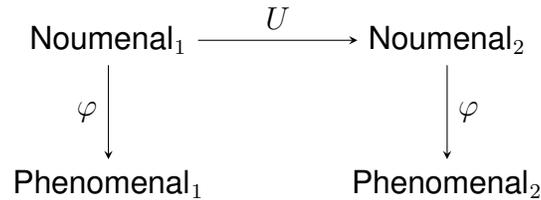
\begin{center}
 \begin{tikzpicture}

      \matrix (m) [matrix of math nodes,row sep=3em,column sep=4em,minimum width=2em,ampersand replacement = \&]
      { \textsf{Noumenal}_{1} \&  \textsf{Noumenal}_{2} \\ \textsf{Phenomenal}_{1} \& \textsf{Phenomenal}_{2} \\};
      \path[-stealth]
      (m-1-1) edge node [left] {$\varphi$} (m-2-1)
              edge node [above] {$U$} (m-1-2)
      (m-1-2) edge node [right] {$\varphi$} (m-2-2);

    \end{tikzpicture}
\end{center}

\noindent
Here, $\varphi$ is a mapping that determines the state of the phenomenal world in function of the state of the noumenal world. We~refer to a  state of the noumenal world as a \emph{noumenal state} and 
to a state of the phenomenal world as a \emph{phenomenal state}.  Any phenomenal state arises from at least one noumenal state. Thus $\varphi$ is \emph{surjective}.

The left part of the picture illustrates the following idea: when the noumenal world is in state $\textsf{Noumenal}_{1}$,  it has a corresponding phenomenal state $\textsf{Phenomenal}_{1}$, which is determined by applying $\varphi$ to $\textsf{Noumenal}_{1}$:
\[ \textsf{Phenomenal}_{1} = \varphi \! \left( \textsf{Noumenal}_{1} \right) . \]

The Law of Nature that determines the evolution of the noumenal world is represented by $U$ in this diagram.
The world is made of \emph{systems} that may or may not interact with each other.
A system is \emph{open} when its evolution can affect the rest of the universe,
and \emph{closed}
when it does not. When it is closed, nothing from the system can escape to the environment.

Consequently, we~can think of $U$ as a transformation that takes as input a noumenal state and outputs a new noumenal state.
This makes sense for the whole universe as well as for any closed system. In~the latter case,
the precise transformation applied would depend on  various factors including time and the state of the \emph{environment}, which is the part of the universe external to the closed system.

The upper part of the picture illustrates the following fact: if we apply a transformation $U$ to a closed system that was in state $\textsf{Noumenal}_{1}$, the new state of the system, $\textsf{Noumenal}_2$, is determined only by its previous state and the transformation. This can be summarized in the following equation:
\[ \textsf{Noumenal}_{2} =  U \star \textsf{Noumenal}_{1}  \, . \]
Note that we wrote ``$U \star \textsf{Noumenal}$'' above, rather than the more familiar form
``$U(\textsf{Noumenal})$'', because we should not think here of $U$ as a \emph{function},
but rather {``$\star$'' is an \emph{action} and $U$ \emph{acts} on the noumenal state according
to that action}. This allows us to use the same $U$ to act differently
on noumenal and phenomenal states by invoking different actions.

Finally, in the right part of the picture, we see that from the new  noumenal state, $\textsf{Noumenal}_{2}$, corresponds a phenomenal state, $\textsf{Phenomenal}_{2}$. Mathematically:
\[ \textsf{Phenomenal}_{2} = \varphi \! \left( \textsf{Noumenal}_{2} \right)  . \]

\subsection{Parallel process between noumenal and phenomenal worlds}\label{subsc:parallel}

A question arises naturally when considering a closed system: 
is~it possible to describe its phenomenal evolution without having recourse to the noumenal world?  Could we explain the evolution of 
phenomenal states only in terms of phenomenal states and transformations applied on them? Could we explain the evolution of state
$\textsf{Phenomenal}_{1}$ to state $\textsf{Phenomenal}_{2} $ through \mbox{transformation} $U$, without
invoking the underling  state $\textsf{Noumenal}_{1}$ giving rise to state $\textsf{Noumenal}_{2}$?

Mathematically, can the following equation be well-defined:
\[  \textsf{Phenomenal}_{2} =  U \cdot \textsf{Phenomenal}_{1}  \, , \]
where we have used ``$\cdot$'' to distinguish this action from the one on noumenal states,
which was denoted ``$\star$'' above?
We now argue that the answer is yes.

The equation will be well-defined whenever  any two possibly   distinct underlying noumenal states $\textsf{Noumenal}_{1}$ and $\textsf{Noumenal}_{1}^{*}$ giving rise to the same state $\textsf{Phenomenal}_{1}$ must also give rise to the same phenomenal states after evolution through the transformation~$U$. 
We now argue that this will be the case.  Suppose we have two states $\textsf{Noumenal}_{1}$ and $\textsf{Noumenal}_{1}^{*}$ giving rise to the same phenomenal
state  $\textsf{Phemonenal}_{1}$. This means that these noumenal states 
are completely indistinguishable
by any experiment whatsoever. This would include experiments that start by applying transformation~$U$.
It~follows that states  $\textsf{Noumenal}_{2} = U \star \textsf{Noumenal}_{1}$ and $\textsf{Noumenal}_{2}^{*} = U \star \textsf{Noumenal}_{1}^{*}$  must also be indistinguishable,
and therefore they cannot give rise to distinct phenomenal states.

Thus, a state $\textsf{Phenomenal}_{1}$, on which a transformation $U$ is applied, will evolve to a well-defined state $\textsf{Phenomenal}_{2}$\,.
Hence, we can  write:
\[ \textsf{Phenomenal}_{2} =  U \cdot \textsf{Phenomenal}_{1}  \, . \]

It follows that there are two ways in which the same state $\textsf{Phenomenal}_{2} $ can be reached from $\textsf{Noumenal}_{1}$.

\begin{itemize}
\item[{$\bullet$}] We can  apply first transformation $U$ to $\textsf{Noumenal}_{1} $ to {obtain} $\textsf{Noumenal}_{2}$  and then apply function $\varphi$ to $\textsf{Noumenal}_{2}$ and determine $\textsf{Phenomenal}_{2}${;}
\item[{$\bullet$}] or we could apply first $\varphi$ to $\textsf{Noumenal}_{1}$ to determine $\textsf{Phenomenal}_{1}$  and then we apply $U$ to $\textsf{Phenomenal}_{1}$ to {obtain} $\textsf{Phenomenal}_{2}$.
\end{itemize}

This allows us to update our diagram to illustrate the parallel process between the evolution of the noumenal world and the phenomenal world:

\begin{center}
 \begin{tikzpicture}

      \matrix (m) [matrix of math nodes,row sep=3em,column sep=4em,minimum width=2em,ampersand replacement = \&]
      { \textsf{Noumenal}_{1} \&  \textsf{Noumenal}_{2} \\ \textsf{Phenomenal}_{1} \& \textsf{Phenomenal}_{2} \\};
      \path[-stealth]
      (m-1-1) edge node [left] {$\varphi$} (m-2-1)
              edge node [above] {$U$} (m-1-2)
      (m-1-2) edge node [right] {$\varphi$} (m-2-2)
      (m-2-1) edge node [below] {$U$} (m-2-2);

    \end{tikzpicture}
\end{center}

This commuting diagram states  that the evolution of the phenomenal consequences of the noumenal world are the phenomenal consequences of the evolution of the noumenal world. 

Mathematically  $\varphi$ is a \emph{homomorphism} that verifies:
\[  \varphi \! \left( U \star \textsf{Noumenal}_{1} \right)  = U \cdot \varphi \! \left( \textsf{Noumenal}_{1} \right)  . \]

The evolution of the phenomenal world  is an \emph{epiphenomenon}: Understanding how the noumenal world evolves and the relation between the noumenal world and the phenomenal world is sufficient to describe  the evolution of the phenomenal world.   

\subsection{Leibniz's Principle}\label{sc:Leibniz}
The previous discussion was made necessary by the possibility of two different noumenal states that can give rise to the same phenomenal state. This possibility runs against a principle attributed to Leibniz~\cite{Indiscernibles,Leibniz},
which claims that if there is no possible perceptible difference between two objects, then these objects are the same, not superficially, but fundamentally. This would imply that if two phenomenal states are equal, then they must arise from the same noumenal state and hence that $\varphi$ is $\emph{injective}$.  As seen previously homomorphism $\varphi$ is surjective,  thus $\varphi$ would be an isomorphism between the noumenal world and the phenomenal world.  Given that an isomorphism is a mere rebranding of terms, the noumenal-phenomenal distinction would be unnecessary since  it  would provide no additional explanatory power.  
However, it can be proven that there are no local-realistic models of quantum theory that satisfy Leibniz's principle\cite{QMlocal}.

\subsection{Principles of Local-Realism}\label{sc:princrealism}
Now that we have discussed the nature of realism  we turn toward the nature of local-realism.
Intuitively, a local-realism can be described by  the following principles:

\paragraph{Principles of Local Realism} 
\begin{enumerate}
\item There is a real world which consists of various parts, called \emph{systems}.
\item A system may be decomposed into subsystems.
\item  Every system is a subsystem of the \emph{global system} consisting of the entire world.
\item At any given time, each system is in some state.
\item The state of a system  determines, and is determined by, the state of its subsystems. 
\item What is observable in a system is determined by the state of the system.
\item The state of the world evolves according to some  law.
\item  The evolution of the state of a system can only be influenced by the state
  of systems in its local neighbourhood.
\end{enumerate}

The notion of local neighbourhood depends on the underlying physics. In the case of relativity theory, this would mean that no systems  can influence each others if they are space-like separated.  This implies that no action on a system can have an influence on  another system at a speed faster than light. However, these principles are not restricted  to  theories in which general relativity holds.

\subsection{Separation beyond relativity theory} \label{subsc:spacelike}

One main lesson of the shift from Newtonian space-time to relativity is this:    there exists events that are not causally related to one another.  Provided Alice and Bob are space-like separated,   
it does not matter  whether Alice act before Bob in a certain reference frame or Bob act before Alice in another one. This is true because in reality, neither is acting before the other and neither is influencing the other in the slightest way.    A generalized characterization of local-realistic theories beyond relativity must be able to incorporate these ideas without reference to the speed of light.

In a local-realistic theory,
provided systems $A$ and $B$ are sufficiently far apart,
it  should not matter if we perform operation $U$ first on system $A$ and nothing on system $B$, followed by nothing on system $A$ and $V$ on system $B$, or if first we do nothing on system $A$ and $V$ on system $B$, followed by 
$U$ on system $A$ and nothing on system $B$.
In~either case, this simply corresponds to performing $U$ on $A$ and $V$ on~$B$.
This is illustrated by the following three circuits, inspired by quantum computational networks~\cite{Deutsch}, whose effect is identical.
\begin{center}
\begin{picture}(370,50) 
\put(0,40){\line(1,0){25}}
\put(25,30){\framebox(20,20){$U$}}
\put(45,40){\line(1,0){65}}
\put(0,10){\line(1,0){65}}
\put(65,0){\framebox(20,20){$V$}}
\put(85,10){\line(1,0){25}}
\put(150,40){\line(1,0){65}}
\put(215,30){\framebox(20,20){$U$}}
\put(235,40){\line(1,0){25}}
\put(150,10){\line(1,0){25}}
\put(175,0){\framebox(20,20){$V$}}
\put(195,10){\line(1,0){65}}
\put(300,40){\line(1,0){25}}
\put(325,30){\framebox(20,20){$U$}}
\put(345,40){\line(1,0){25}}
\put(300,10){\line(1,0){25}}
\put(325,0){\framebox(20,20){$V$}}
\put(345,10){\line(1,0){25}}
\end{picture}
\end{center}
Simply put, it is not meaningful to say that $U$ was done before $V$ or vice versa. In a  local-realistic theory, one key property is that  the state of system $A$ should only change when operation $U$ is done to it, independently of whether operation $V$ was done or not on system $B$.

\section{Local-realistic theories}

\subsection{Axioms and models}
The concept of a local-realistic theory will soon be formalized according to the principles enumerated of section  \ref{sc:princrealism}. 
This will be done through the creation of an explicit axiomatic system that characterize the notion of  a local-realistic theory.

Intuitively, an \emph{axiomatic system} is simply a list of mathematical propositions called \emph{axioms}  expressed in the language of predicate logic.  Here, for ease of reading, the axiomatic system of local-realistic theories and no-signalling theories is expressed through the use of English supplemented by some mathematical symbols, rather than in the recondite language of quantifiers, logical connective and predicates.  
A mathematical proposition in an axiomatic system also contains various  mathematical symbols that do not  refer to a specific mathematical-objects.  For example, the axioms of local-realistic theories  contain   symbols like  ``$\pi_{A}^{B}$'' or ``$\textsf{Noumenal-Space}^{A}$'' and the axioms do not tell us to what objects these symbols refer to.
A \emph{model} of an axiomatic system is an assignment to every  such  mathematical symbol to a specific mathematical object, in such  a way that all the axioms are true.  Formal definitions of axiomatic system and model are given in ref.\cite{structure}.
 
A mathematical theory is defined as  local-realistic if it  satisfies \emph{all} the axioms of local-realistic theories that shall be presented here.  
In contrast, a physical theory should not be seen as a purely mathematical theory, since it also contains non-mathematical entities related to its interpretation and to other factors.  The precise question of when a physical theory should be considered local-realistic shall be discussed in the conclusion of this paper.

\subsection{Systems}
We want to define the mathematical properties of systems, where systems describe meaningful parts of the world.  These systems can be combined in various way to give rise to different systems.  For example, if $A$ is a system, there will be a system $\overline{A}$, the complement of system $A$, consisting of the rest of the world.  As an other example, if $A$ and $B$ are systems, there will be a system $A \sqcup B$, the union of system $A$ and $B$ consisting of the parts of the world consisting of the part of the world belonging to  either of the two systems $A$ and $B$.

Mathematically,  systems will be represented as elements of a lattice of systems, which we define now.

\begin{definition}[Lattice of systems]  
A lattice of systems is a 6-tuple $\left( \mathcal{S}, \sqcup, \sqcap, \overline{\,\cdot\,}, S, 0 \right)$, where $\mathcal{S}$ is a set of elements called \emph{systems}.

There are two special systems:

\begin{enumerate} 
\item $S$, which is the \emph{whole system} being considered, hereinafter called the \emph{global system}. It~could be the entire universe.  Alternatively, it could be something much smaller, like a quantum computer or a single photon.  
\item The \emph{empty system} $0$, which contains no parts at all.
\end{enumerate}

Let $A$ and $B$ be systems, then:

\begin{enumerate}
\item There exists a system $A \sqcup B$, the \emph{union} of  $A$ and $B$.
\item There exists a system $A \sqcap B$ the \emph{intersection} of  $A$ and $B$.
\item  There exists a system $\overline{A}$, the \emph{complement} of $A$, which
is defined so that $A \sqcap \overline{A}=0$ and $A \sqcup \overline{A}=S$.
Intuitively, it is composed of all the parts of $S$ that are not in $A$.
\end{enumerate}

The transformations ($\sqcup$, $\sqcap$, $\overline{\,\cdot\,}$\,) and distinguished elements ($S$, $0$) behave like their usual set-theoretic counterparts. We use the slightly different notation of $\sqcup$, $\sqcap$, rather than $\cup$, $\cap$, to emphasize the fact that the transformations $\sqcup$ and $\sqcap$ are purely alge\-braic in nature.
Formally, a lattice of systems is a \emph{boolean lattice}.
\end{definition}

Given the definition of a  lattice of systems, we can state the first axiom.

\begin{axiom}[Systems] \label{ax:systems}
Associated to a local-realistic theory is a lattice of systems  $\left( \mathcal{S}, \sqcup, \sqcap, \overline{\,\cdot\,}, S, 0 \right)$.
\end{axiom}

We now introduce some terminology on systems.

\begin{definition}[Subsystem]
System $A$ is a \emph{subsystem} of a system~$B$, written \mbox{$A \sqsubseteq B $}, if $A \sqcap B = A$.
\end{definition}

\begin{definition}[Disjoint systems]
Systems $A$ and $B$ are \emph{disjoint} if they have no parts in common, i.e.~$A \sqcap B =0 $. 
\end{definition}

Note that the empty system is a subsystem of all systems, including itself, and that it is disjoint from all
systems, again including itself.

\begin{definition}[Composite system]
Let $A$ and $B$ be disjoint.  The system $A \sqcup B$ is a \emph{composite system}, composed of systems $A$ and $B$.  For convenience, we denote it by $AB$, rather than $A\sqcup B$.
\end{definition}

Since $\sqcup$ is commutative, we have that
\mbox{$AB = A \sqcup B = B \sqcup A = BA$}.

Since $\sqcup$ is also associative, we have
\mbox{$A \left( BC \right) = \left( A B \right) C$}
for any three mutually disjoint systems $A$, $B$ and $C$.
Thus, we shall simply write $ABC$ to denote the composite system consisting of $A$, $B$ and $C$.

\subsection{States}

We follow Kant's terminology~\cite{Kant-Brit, Kant}, and thus we distinguish two kinds of states in a system, as mentioned informally in Section~\ref{sc:concepts}. 

\begin{description}
\item[Noumenal State:] The noumenal state of a system is its complete description. It~describes the system as it is, rather than what can be observed about it, or known through sensory experience.  It describes not only what can be observed from a system, but also 
how the system can interact with other systems. It is a state of being.  It describes the system in itself, including parts that are not observable locally or at all. Another term used in quantum foundations literature to describe the noumenal state  would be the \emph{ontic} state~\cite{toy}.
\item[Phenomenal State:] The phenomenal state of a system is a complete description of what is \emph{locally} observable in {that} system.  The phenomenal state is a complete description of all the observable properties potentially accessible in a system.  It is what is observable in a system; not what is actually observed. The phenomenal state contains everything that can be observed through arbi\-trarily powerful technology. The only restriction on the technology is that it must abide by the laws of Nature. 
\end{description} 

The choice of terminology reflects the difference between appearance and reality. An~alter\-na\-tive
distinction, which is somewhat orthogonal, concerns the difference \mbox{between} existence and knowledge.  The theories of existence and of knowledge are dealt with in the respective branches of philosophy called \emph{ontology} and \emph{epistemology}. Following that path would have led to the distinction between the \emph{ontic} state of a system and its \emph{epistemic} state~\cite{toy}. The ontic state corresponds to what we have called the noumenal state.  However, the epistemic state corresponds to what is known about a system by some observer~\cite{LeifSpekA}, which might be subjective and vary from one observer to another~\cite{Qubism,LeifSpekB}.
It~should be emphasized that our phenomenal states are \emph{not} states of knowledge, neither are they relative to an observer.  
Hence, epistemic and phenomenal states are two fundamentally different notions.

This leads us to corresponding  axioms:

\begin{axiom}[Noumenal state space]
Associated to a system $A$ is  \emph{noumenal state space}, $\textsf{Noumenal-Space}^{A}$,  which is a non-empty set of \emph{noumenal states}.
\end{axiom}

 Particular noumenal states of $A$ will be  denoted $N^{A}$,  $ N^{A}_i$, $ N^{A}_1$, etc. 

Intuitively, the noumenal state space of a system consists of all possible noumenal states that it could theoretically be in.

\begin{axiom}[Phenomenal state space]\label{pheno}
Associated to a system $A$  is a \emph{phenomenal state space}, $\textsf{Phenomenal-Space}^{A}$, which is a non-empty set of \emph{phenomenal states}.
\end{axiom} 

Particular phenomenal states of $A$ will be denoted $\rho^{A}$, $\rho^{A}_{i}$, $\rho^{A}_{1}$, etc.

Intuitively, the phenomenal state space of a systems consist of all possible phenomenal state that it could theoretically be in.

\subsection{Transformations and actions}
We now want to describe how states can evolve in a local-realistic theory.
The evolution of a system happens through transformations that can be applied to the system.  The precise transformation that is applied  might depend on: the environment external to the global system which might be empty if the global system is the whole universe, on the dynamical laws of physics; or on time and other variables.

The next definition  characterizes the mathematical properties of a set of transformations  and how these transformations may be combined into to new transformations.

\begin{definition}[Monoid of transformations]  A monoid of transformations is a 3-tuple $ ( \textsf{Transformations}, \circ , I ) $, where $\textsf{Transformations}$ is a set whose elements are called \emph{transformations}. The set of transformations comes with a binary operator denoted~``$\circ$'', called the \emph{composition}, and $I$ is a transformation called the \emph{identity transformation}.  A monoid of transformations satisfies the following properties:
\begin{enumerate}
\item If $U$ and $V$ are transformations, $U \circ V$ is a transformation called the \emph{composite} of $U$ and $V$;
\item If $U$, $V$ and $W$ are transformations, $U \circ  \left( V \circ  W \right) = \left( U \circ V \right) \circ W $;
\item  For all transformations $U$,
\[I \circ U = U \circ I = U \, . \]
\end{enumerate}

Importantly, any transformation that might be physically applied to a system must belong to the set of transformations associated to that system, but the converse might not hold. In quantum theory, the transformations associated to a system consist of the unitary operations of the Hilbert space associated to that system, but the transformations physically possible might be   further restricted by other considerations. For example when working in quantum computing, we might declare that only unitary operations  that are finitely generated by the composition  of a certain set of quantum gates are physically realizable, while other unitary operations are merely meaningful mathematical transformations.    

\noindent
When there is no ambiguity, we shall omit the composition operator and write $UV$ instead of $U \circ V$.

\end{definition}

When the transformations in a monoid are invertible,  they define a group of transformations as follows.

\begin{definition}[Group of transformations]
A monoid of transformation $( \textsf{Transformations}, \circ , I ) $ is a \emph{group} of transformation if associated to every transformation $U$ in $\textsf{Transformations}$ there exist a transformation $V$, the \emph{inverse} of $U$, which has the property that  $U \, V = V \, U = I$.
\end{definition} 
In a group of transformation, the inverse of a transformation $U$ is unique and will be denoted $U^{-1}$.

Now:

\begin{axiom}[Transformations on a system]\label{op}
Associated to a system  $A$   is a monoid  of transformations, $(\textsf{Transformations}^{A} , \circ^{A}, I^{A} ) $.
\end{axiom}

Particular transformations on system $A$ are denoted $U^{A}$, $V^{A}$, etc. Also, $I^{A}$ denotes the identity transformation on system $A$.  When there is no ambiguity, we drop the superscript and write simply $U$, $V$ and~$I$.

Intuitively, the transformations associated to a system are the transformations that can be performed on the system, and if $U$ and $V$ are transformations on a system, the transformation $U \circ V$ can be implemented by first doing $V$, followed by doing $U$.

A transformation can act on a state to produce a new state, as follows:

\begin{definition}[Action]
Let  $(  \textsf{Transformations}, \circ, I ) $ be a monoid of transformations and $S$ be a set.
An~\emph{action} of the monoid of transformation on set $S$ is a binary operator \mbox{$\star : \textsf{Transformations} \times S \to S$} that satisfies, for all transformations $U$ and $V$ and for all element $s$ of the set $S$,
\begin{enumerate}
\item $ U \star \left( V \star s \right) = \left( UV \right) \star s $\,;
\item  $I \star s = s$.
\end{enumerate}
\end{definition}

The next two  axioms use the above definitions to express the fact that a transformation done on a system changes its underlying noumenal and phenomenal state.

\begin{axiom}[Noumenal action]\label{axiom:noumenalaction}
Associated to a system $A$, is a \emph{noumenal action}  denoted ``$\star^{A}$'', which is an action of the monoid of transformations of  the system on the set of noumenal states of the system.
\end{axiom}
When there is no ambiguity, we drop the superscript of the noumenal action.  For example, we write  $U \star N^{A}$ instead of $U \star^{A} N^{A}$.

Intuitively, if  system $A$ was in noumenal state $N^{A}$, and a transformation $U$ is done on it, its new noumenal state is $U \star N^{A}$.

\begin{axiom}[Phenomenal action]\label{actionpheno}
Associated to a system $A$, is a \emph{phenomenal action} ``$\cdot^{A}$'', which is an action of the monoid of transformations of the system on the set of phenomenal states of the system.
\end{axiom}

When there is no ambiguity, we drop the superscript of the phenomenal action. For example, we write $U \cdot \rho^{A}$ instead of $U \cdot^{A}\rho^{A} $.

Intuitively, if a system $A$ was in phenomenal state $\rho^{A}$, and a transformation $U$ is done on system $A$, its new phenomenal state is $U \cdot \rho^{A}$.

Sometimes, transformations can be characterized precisely in terms of how they act on a given set. 
This leads to the concept of a \emph{faithful} action.
\begin{definition}[Faithful action]\label{def:faithful}
Let $ \star $  be an action of a monoid of transformations  on a set $S$.  
The~action is \emph{faithful} if  transformations $U$ and $V$ are equal  whenever $U \star  s$ is equal to  $V \star s$ for every element $s$ of  $S$.  Thus, the action is faithful  if whenever two transformations act identically on all elements of $S$ they are equal.
\end{definition}

The next axiom states that if two transformations act identically all noumenal states, then they are equal.

\begin{axiom}[Noumenal faithfulness] \label{ax:faithnoume}
For every system, its associated  noumenal action is faithful.
\end{axiom}
It is both algebraically very useful and physically natural to impose noumenal faithfulness. However,
this axiom is not fundamental because any theory that verifies all axioms of local-realism except this one  can be transformed into a noumenally faithful local-realistic theory by replacing transformations by equivalence classes of transformations, in effect equating any two transformations that act identically on all possible noumenal states. For details, see  appendix~\ref{sc:faithfulness}.

\subsection{Noumenal-phenomenal homomorphism}
In a local-realistic theory, what is observable locally in a system is determined by the complete description of that system, in other words, the noumenal state of a system determines its phenomenal state.  If the noumenal state of a system evolves according to a transformation, its corresponding phenomenal state must evolve according to the same transformation.
 Mathematically the phenomenal state of a system will be determined by its underlying noumenal state, through a structure-preserving surjective map -- a noumenal-phenomenal epimorphism. But first:

\begin{definition}[Noumenal-phenomenal homomorphism] Let $A$ be a system and let $\phi$ be a mapping whose domain is the noumenal state space of $A$ and whose range is the phenomenal state space of $A$.  We~say that $\phi$ is a \emph{noumenal-phenomenal homomorphism} on system $A$ if, for any transformation $U$  of system $A$ and any noumenal state $N$ of  $A$,
\[  \phi \! \left(  U \star N \right) = U  \cdot \phi \! \left( N \right ) \, . \]  
\end{definition}

When no ambiguity can arise, we omit  the actions.  For example, the equation above becomes
\[  \phi \! \left(  U   N  \right) = U   \,  \phi \! \left(  N \right)  \,  . \]

\begin{definition}[Noumenal-phenomenal epimorphism] A surjective noumenal-phenomenal homomorphism on system $A$  is called a \emph{noumenal-phenomenal epimorphism} on system $A$.
\end{definition}

\begin{axiom}[Noumenal-phenomenal epimorphism]  \label{ax:noupheepi}
Associated with each system $A$ is  a noumenal-phenomenal epimorphism denoted~$\varphi^{A}$ called \emph{the} noumenal-phenomenal epimorphism of system $A$.
\end{axiom}
When there is no ambiguity, we write $\varphi$ \mbox{instead} of~$\varphi^{A}$. 

Intuitively, if a system $A$ is in noumenal state $N^{A}$, it has a corresponding phenomenal state $\rho^{A} = \varphi ( N^{A} )$. Furthermore, if a transformation $U$ is done on system $A$, its new noumenal state is $ U  N^{A}$, and its corresponding new phenomenal state is $U \rho^{A}$.  Lastly, every phenomenal state on a system arises from at least one noumenal state, since what is observable has an underlying reality.

The transformations act in a way that leads to the parallel evolution of the noumenal world and the phenomenal world, as explained intuitively in Section~\ref{subsc:parallel}.
This is best illustrated by the commuting diagram that we had seen previously:

\begin{center}
 \begin{tikzpicture}

      \matrix (m) [matrix of math nodes,row sep=3em,column sep=4em,minimum width=2em,ampersand replacement = \&]
      { N_{1} \&  N_{2} \\ \rho_{1} \& \rho_{2} \\};
      \path[-stealth]
      (m-1-1) edge node [left] {$\varphi$} (m-2-1)
              edge node [above] {$U$} (m-1-2)
      (m-1-2) edge node [right] {$\varphi$} (m-2-2)
      (m-2-1) edge node [below] {$U$} (m-2-2);

    \end{tikzpicture}
\end{center}

\subsection{Splitting and merging}\label{sc:splitandmerge}
As we explained informally in Section~\ref{sc:princrealism}, a local-realistic world can be decomposed into several parts.
These parts exist in such a way that the state of the whole determines the state of the parts, and conversely the state of the whole is fully determined by the state of the parts.
Note that the latter is \emph{not} the case with the standard description of  
quantum theory
since entangled states cannot be recovered from the state of their subsystems. 
This is the reason why the usual formalism does not provide a local-realistic model of quantum theory.

Given a composite system $AB$, its noumenal state $N^{AB}$ can be decomposed in two states:  A~noumenal state $N^{A}$, in the state space of $A$, and a noumenal state $N^{B}$, in the state space of~$B$. Informally, the state
of the parts is determined by the state of the whole.
For this purpose, we shall introduce formally with Axiom~\ref{ax:noumeproj} two \emph{projectors}, $\pi_{A}$ and $\pi_{B}$, which split a system in the following way:
\[ N^{A} = \pi_{A} \! \left( N^{AB} \right) ~~\text{and}~~
  N^{B} = \pi_{B} \! \left( N^{AB} \right) . \]

Furthermore, the two noumenal states $N^{A}$ and $N^{B}$ determine completely the noumenal state $N^{AB}$.
Informally, the state of the whole is determined by the state of the parts. 
For this purpose, we shall introduce formally with Axiom~\ref{ax:join} a \emph{noumenal product} ``$\odot$'', which merges the noumenal states of systems $A$ and~$B$ as follows:
\[ N^{AB} = N^{A} \odot N^{B} \, . \]

This is illustrated by the following diagram.
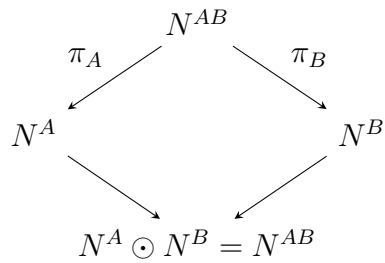
\begin{center}
    \begin{tikzpicture}
      \matrix (m) [matrix of math nodes,row sep=2em,column sep=0em,minimum width=1em,ampersand replacement = \&]
      {  \&  N^{AB} \&  \\ N^{A} \&  \& N^{B}  \\ \& N^{A} \odot N^{B}  = N^{AB}  \& \\ };
      \path[-stealth]
      (m-1-2) edge node [above left] {$\pi_{A}$} (m-2-1)
       (m-1-2)       edge node [above right] {$\pi_{B}$} (m-2-3)
      (m-2-3) edge node [below] {}  (m-3-2)
      (m-2-1) edge    (m-3-2);
    \end{tikzpicture}
\end{center}

Note that such a diagram would not be possible at the phenomenal level in quantum theory,
if~we replaced $N$ by~$\rho$ throughout.
Nevertheless, even though the phenomenal state $\rho^{AB}$ of composite system $AB$ cannot be determined
from the phenomenal states $\rho^{A}$ and $\rho^{B}$ of subsystems $A$ and~$B$,
it~\emph{can} be determined (as~well as $\rho^{A}$ and $\rho^{B}$) from the \emph{noumenal} states
$N^{A}$ and $N^{B}$ of $A$ and~$B$, as illustrated by the following diagram.

\begin{center}
    \begin{tikzpicture}
      \matrix (m) [matrix of math nodes,row sep=1em,column sep=0em,minimum width=1em,ampersand replacement = \&]
      {   N^{A} \&  \& N^{B}  \\ \& N^{A} \odot N^{B} = N^{AB} \& \\ \rho^{A} \& \& \rho^{B} \\ \& \rho^{AB} \& \\ };
      \path[-stealth]
      (m-1-1) edge   (m-2-2)
      (m-1-3) edge    (m-2-2)
      (m-1-1) edge node [left] {$\varphi$} (m-3-1)
      (m-1-3) edge node [right] {$\varphi$} (m-3-3)
      (m-2-2) edge node [right] {$\varphi$} (m-4-2);
    \end{tikzpicture}
\end{center}

\subsubsection{Noumenal and phenomenal projectors}

Let us now proceed formally.

The next axioms express the fact that the noumenal state of a system determines the noumenal state of any of its subsystems.

\begin{axiom}[Noumenal projector]\label{ax:noumeproj}
Associated to  all systems $A$ and $B$ such that $A$ is a subsystem of $B$, is a function denoted $\pi_{A}^{B}$, which is called the \emph{ noumenal projector} from system $B$ to system $A$.
Projector $\pi_{A}^{B}$ is a surjective function from the noumenal space of system $B$ to the noumenal space of system $A$.

Furthermore, for all systems $A$, $B$ and $C$, where $A$ is a subsystem of $B$, which is itself a subsystem of $C$ then the following relation hold between projectors:
\[ \pi_{A}^{B} \circ \pi_{B}^{C}= \pi_{A}^{C} , \]
where $\circ$ denotes the composition of functions.  
\end{axiom}

When there is no ambiguity, we shall omit the superscript  and we shall refer to $\pi_{A} $ as \emph{the}
noumenal projector  to system~$A$, regardless of the supersystem from which it is projected. 
For example, the previous equation will simply be written as:
\[  \pi_{A} \circ \pi_{B} = \mathrm{\pi}_{A} \, . \]

Intuitively if a  system $B$ is in noumenal state $N^{B}$, the noumenal state of a subsystem $A$ will be $N^{A} = \pi_{A} ( N^{B} ) $.

The next axiom expresses the fact that the phenomenal state of a system determines the phenomenal state of any of its subsystems.  

\begin{axiom}[Phenomenal projector]\label{tracepheno}
Associated to  all systems $A$ and $B$ such that $A$ is a subsystem of $B$, is a function  called the \emph{phenomenal projector} from system $B$ to system $A$. These phenomenal projectors follow the same requirements as noumenal projectors, as stated in axiom~\ref{ax:noumeproj},
\emph{mutatis mutandis}.  As an abuse of notation, we also \mbox{denote} the phenomenal projectors from system $B$ to system $A$ by $\pi_{A}^{B}$, since no ambiguity will be possible with the corresponding noumenal projector $\pi_{A}^{B}$.
\end{axiom}
Finally, when there is no ambiguity, we shall omit the superscript  and we shall refer to $\pi_{A} $ as \emph{the}
phenomenal projector  to system~$A$, regardless of the supersystem from which it is projected.

Intuitively if a  system $B$ is in phenomenal state $\rho^{B}$, the noumenal state of a subsystem $A$ will be $\rho^{A} = \pi_{A} ( \rho^{B} ) $.

Projectors are idempotent \footnote{\,By definition, $x$ is idempotent
under operation ``$\cdot$'' when $x \cdot x = x$.} under composition, i.e.:
\[ \pi_{A} \circ \pi_{A} = \pi_{A} \, . \]
This lead to the following two theorems.

\begin{theorem}
For any noumenal state $N^{A}$ of system $A$,
\[ \pi_{A} \! \left( N^{A} \right) = N^{A} \, .\]
\end{theorem}
\begin{proof}
The surjectivity of $\pi_{A}^{A}$ implies that there must exist a state $N_{\alpha}^{A} $ in the noumenal space of system $A$ such that $ \pi_{A} \! \left( N_{\alpha}^{A} \right) = N^{A} $.
Therefore,
\[ \pi_{A} \! \left( N^{A} \right) =  \pi_{A} ( \pi_{A} ( N^{A}_{\alpha} ) ) = ( \pi_{A} \circ \pi_{A} )  \! \left( N_{\alpha}^{A} \right)   = \pi_{A} \! \left( N_{\alpha}^{A} \right) = N^{A} \, . \]
Of course, this $N_{\alpha}^{A}$ is none other than the original $N^{A}$ since
$\pi_{A} \! \left( N_{\alpha}^{A} \right)  = N^{A}$ by definition of $N_{\alpha}^{A} $,
but also
$\pi_{A} \! \left( N_{\alpha}^{A} \right)  = N_{\alpha}^{A}$ by the theorem itself.
\end{proof}

\begin{theorem}
For any phenomenal state $\rho^{A}$ of system $A$,
\[ \pi_{A} \! \left( \rho^{A} \right) = \rho^{A} \, .\]
\end{theorem}
\begin{proof}
The proof is left to the reader.
\end{proof}

The next axiom expresses the consistency between the noumenal and phenomenal projectors.
\begin{axiom}[Relation between noumenal and phenomenal projectors] \label{ax:relnouphe}
For all systems  $A$ and $B $ such that $A$ is a subsystem of $B$, and all noumenal states $N^{B} $, the noumenal and phenomenal projections to system $A$ are related by the following commuting relation:
\[
\pi_{A} \! \left( \varphi\! \left( N^{B} \right) \right) = \varphi \! \left( \pi_{A} \! \left( N^{B} \right) \right) .
\]
Here the symbol $\varphi$ stands for $\varphi^{AB} $ on the left side of the equation, but for $\varphi^{A} $ on the right side.  Also the symbol $\pi_{A}$ stands for the phenomenal projection on the left side of the equation, but for the noumenal projection on the right side. 
\end{axiom}

The relation between the noumenal and phenomenal projectors is best visualized by the fact that the following diagram commutes.

\begin{center}
 \begin{tikzpicture}

      \matrix (m) [matrix of math nodes,row sep=3em,column sep=4em,minimum width=2em,ampersand replacement = \&]
      { N^{B} \&  N^{A} \\ \rho^{B} \& \rho^{A} \\};
      \path[-stealth]
      (m-1-1) edge node [left] {$\varphi$} (m-2-1)
              edge node [above] {$\pi_{A}$} (m-1-2)
      (m-1-2) edge node [right] {$\varphi$} (m-2-2)
      (m-2-1) edge node [below] {$\pi_{A}$} (m-2-2);

    \end{tikzpicture}
\end{center}

\subsubsection{Abstract trace}\label{sc:abstract:trace}
Quantum theory often mentions tracing out other systems. 
More generally, we can define an abstract trace from any projector.  For all disjoint systems $A$ and $B$, for all noumenal states $N^{AB}$ and all phenomenal states $\rho^{AB}$, we define
\[ \mathrm{tr}_{B} \! \left(N^{AB} \right) \isdef \pi_{A} \! \left( N^{AB} \right) ~~\text{and}~~
   \mathrm{tr}_{B} \! \left(\rho^{AB} \right) \isdef \pi_{A} \! \left( \rho^{AB} \right) . \]
Again, while both traces are different functions, they will be denoted with the same symbols since  no ambiguity can arise. 
The choice of working with projectors rather than traces stems from the fact that the notion of trace belongs to linear algebra only, whereas projectors are universal mathematical objects.  

\subsubsection{Noumenal Product}

We shall need the following definition to express the next axiom.

\begin{definition}[Compatible states]
Let $AB$ be a composite system, let $N^{A}$ and $N^{B}$  be noumenal states of system $A$ and $B$ respectively.  We say that $N^{A} $ and $N^{B}$ are \emph{compatible states} if there exist a noumenal state $N^{AB}$ of system $AB$ such that $N^{A} = \mathrm{tr}_{B} ( N^{AB} ) $ and $ N^{B} = \mathrm{tr}_{A} ( N^{AB} ) $.
\end{definition}

The next axiom  expresses the requirement  that the state of a system is determined by the state of its subsystems.

\begin{axiom}[Noumenal product]\label{ax:join}
Associated to all  pair of disjoint systems $A$ and $B$ is a transformation, the \emph{noumenal product}, denoted  ``$\odot_{ A, B  } $'',
\newcounter{fnjoin}%
\setcounter{fnjoin}{\thefootnote}%
such that for all noumenal states $N^{AB} $ of $AB$,  the following equation holds:
\[ N^{AB} =  \pi_{A} \! \left( N^{AB} \right) \odot_{ A, B  }   \pi_{B} \! \left( N^{AB} \right) \, . \]
Furthermore, $N^{A} \odot N^{B}$ is defined only  if states $N^{A}$ and $N^{B}$ are compatible.
\end{axiom}

When there is no ambiguity, we shall drop the subscript from the noumenal product and write simply $ \odot$ instead of $\odot_{ A, B  } $.

There is no corresponding axiom  for a   phenomenal product, and this  is the most profound distinction between the noumenal and phenomenal level.

\begin{theorem}\label{th:noumeproj}  When $N^{A}$ and $N^{B}$ are compatible:
\[ \pi_{A} ( N^{A} \odot N^{B} ) = N^{A} \]
and 
\[ \pi_{B} ( N^{A} \odot N^{B} ) = N^{B} \]
\end{theorem}
\begin{proof}
We shall prove that $\pi_{A} (  N^{A} \odot N^{B} ) = N^{A} $, the other statement is similar.
Given that $N^{A} $ and $N^{B}$ are compatible, there exist a noumenal $N^{AB}$ of system $AB$ such that $\pi_{A} ( N^{AB} ) = N^{A} $ and $\pi_{B} ( N^{AB} ) = N^{B}$. 
It follows from the axiom of the noumenal product that
\[ \pi_{A} (  N^{A} \odot N^{B} )  = \pi_{A} (   \pi_{A} ( N^{AB}  ) \odot \pi_{B} ( N^{AB} )  ) = \pi_{A} ( N^{AB} ) = N^{A} \,  . \] 
\end{proof}

\begin{theorem}\label{th:fundnoumeprod} Let $AB$ be a composite system,  let $N^{A}$, $N^{B}$ and $N^{AB}$ be noumenal states of systems $A$, $B$ and $AB$, respectively.
\[  N^{A} \odot N^{B} = N^{AB} \]
\[ \iff \]
\[ N^{A} = \pi_{A} ( N^{AB})  \; \text{ and } \; N^{B} = \pi_{B} ( N^{AB} ) \]
\end{theorem}
\begin{proof}
This theorem follows directly from Theorem \ref{th:noumeproj} and Axiom \ref{ax:join}.
\end{proof}
It follows from Theorem \ref{th:fundnoumeprod} that for  an arbitrary noumenal state $N^{AB}$ of composite system $AB$, there exist a unique noumenal state $N^{A}$   of $A$ and  a  unique noumenal state $N^{B} $ of $B$  such that $N^{AB} = N^{A} \odot N^{B}$.  Of course, these states are respectively $\pi_{A} ( N^{AB} ) $ and $\pi_{B} ( N^{AB} ) $.

\begin{theorem}[Commutativity of the noumenal product]\label{th:commu:join}
For any disjoint systems $A$ and~$B$,
\[ N^{A} \odot_{ (A, B ) }  N^{B} = N^{B} \odot_{ ( B, A)}  N^{A} \, .\]
\end{theorem}
\begin{proof}
\begin{align*}
& ~\! \left( N^{A} \odot_{ (A, B ) }  N^{B} \right)  \\
= & ~ \pi_{B} (  N^{A} \odot_{( A, B )}  N^{B} ) \odot_{ ( B, A ) } \pi_{A} ( N^{A} \odot_{ ( A, B ) } N^{B} )   \\
= & ~N^{B} \odot_{ ( B, A ) } N^{A} \, .  \\ \\[-7ex]
\end{align*}
\end{proof}

We are now interested in  showing that the noumenal product is associative, as demonstrated by Theorem \ref{th:assoc:join}.  To do we shall need  the four following lemmas.

\begin{lemma}\label{th:noumeprojasso1} 
 Let $N^{A}$, $N^{B}$ and $N^{C}$ be noumenal states of systems $A$, $B$ and $C$ respectively such that the product $N^{ABC} = ( N^{A} \odot  N^{B}) \odot N^{C} $ is  defined.  The following hold
\[ N^{A} \pi_{A} (  N^{ABC} )  \; \text{ and } \;  N^{B} = \pi_{B} (  N^{ABC} )  \; \text{ and } \;  N^{C} =  \pi_{C} ( N^{ABC} )  \, . \]
\end{lemma}
\begin{proof}
We shall prove that $N^{A} $ is equal to $\pi_{A} ( N^{ABC}  )  $, the other statements are similar and left to the reader.
\begin{align*}
 N^{A} &= \pi_{A} ( N^{A} \odot N^{B} ) \\
                               & = \pi_{A} ( \pi_{AB} ( (N^{A} \odot N^{B}) \odot N^{C} )) \\
 & = \pi_{A} ( \pi_{AB} ( N^{ABC} ) ) \\
& =  \pi_{A} ( N^{ABC} )  \\[-7ex]
\end{align*}
\end{proof}

\begin{lemma}\label{th:noumeprojasso2}  Let $N^{A}$, $N^{B}$ and $N^{C}$ be noumenal states of systems $A$, $B$ and $C$ respectively such that the product $N^{ABC} = N^{A} \odot ( N^{B}) \odot N^{C}) $ is  defined.  The following hold
\[ N^{A} =\pi_{A} (  N^{ABC} )  \; \text{ and } \;  N^{B} = \pi_{B} (  N^{ABC} )  \; \text{ and } \;  N^{C} =  \pi_{C} ( N^{ABC} )  \, . \]
\end{lemma}
\begin{proof}
The proof is similar to Lemma $\ref{th:noumeprojasso1} $ and is left to the reader.
\end{proof}

\begin{lemma}\label{th:noumepro4} Let $N^{ABC}$ be a noumental state of composite system $ABC$, then
\[ \pi_{A} ( N^{ABC} ) \odot \pi_{B} ( N^{ABC} )  = \pi_{AB} ( N^{ABC}  ) \, . \]
\end{lemma}
\begin{proof}
\begin{align*}
& \pi_{A} ( N^{ABC} ) \odot \pi_{B} ( N^{ABC} ) \\
 = & (\pi_{A} \circ \pi_{AB} )( N^{ABC} ) \odot ( \pi_{B} \circ \pi_{AB} ) (N^{ABC} ) \\
  = & (\pi_{A} ( \pi_{AB} ( N^{ABC} )   ) \odot  \pi_{B}  ( \pi_{AB} ( N^{ABC} )    ) \\
  = &   \pi_{AB} ( N^{ABC} ) \, .  \\[-7ex]                                                         
\end{align*}
\end{proof}

\begin{lemma}\label{th:noumeprojasso3}Let $N^{ABC}$ be a noumenal state of composite system $ABC$, let $N^{A} = \pi_{A} ( N^{ABC} ) $, $N^{B} = \pi_{B} ( N^{ABC}  )$, $N^{C}= \pi_{C}  ( N^{ABC} )$.  The following hold:
\[ ( N^{A} \odot  N^{B} ) \odot N^{C}  = N^{ABC} =  N^{A} \odot ( N^{B} \odot N^{C} )  \, . \] 
\end{lemma}
\begin{proof}
Importantly by Lemmas  \ref{th:noumeprojasso1} and \ref{th:noumeprojasso2}, both products $(N^{A} \odot N^{B} ) \odot N^{C} $ and $N^{A} \odot (N^{B} \odot N^{C} ) $ are defined.
  
We now show that $( N^{A} \odot  N^{B} ) \odot N^{C}$ is equal to   $ N^{ABC}$.  We leave it to the reader to verify that  $  N^{A} \odot ( N^{B} \odot N^{C} )$ is also equal to $N^{ABC}$.  Using Lemma \ref{th:noumepro4}, we can see that
\begin{align*}
 &  ( N^{A} \odot  N^{B} ) \odot N^{C}   \\
=& ( \pi_{A} ( N^{ABC} ) \odot \pi_{B} ( N^{ABC} )  ) \odot \pi_{C} ( N^{ABC} ) \\
= &  \pi_{AB} ( N^{ABC} ) \odot \pi_{C} ( N^{ABC} )  \\
= &                    N^{ABC} \, . \\[-7ex]
\end{align*}
\end{proof}

\begin{theorem}[Associativity of the noumenal product]\label{th:assoc:join}
Let $ABC$ be a composite system, let $N^{A}$, $N^{B}$, $N^{C}$ be noumenal states of systems $A$, $B$ and $C$ respectively.  The product \mbox{$ (N^{A} \odot N^{B} )\odot N^{C}$} is defined precisely when $N^{A} \odot  ( N^{B} \odot N^{C}) $ is defined, in which case,
\[ \left( N^{A} \odot N^{B} \right) \odot N^{C} = N^{A} \odot \left( N^{B} \odot N^{C} \right) \, . \]
\end{theorem}
\begin{proof}
The theorem follows immediately from Lemmas \ref{th:noumeprojasso1} , \ref{th:noumeprojasso2} and \ref{th:noumeprojasso3}.   
\end{proof}

It follows from the associativity of the noumenal product that  we can omit unnecessary parentheses and  write $N^{A} \odot N^{B} \odot N^{C} $ to denote the noumenal product of $N^{A}$, $N^{B}$ and $N^{C}$.

\begin{theorem}\label{th:assomain} Let $N^{A}$, $N^{B}$, $N^{C}$ and $N^{ABC}$ be noumenal states of systems $A$, $B$, $C$ and $ABC$, respectively.
The product   $ N^{A} \odot N^{B}  \odot N^{C}$ is defined and equal to  $N^{ABC}$ precisely whenever
\[ N^{A} = \pi_{A} ( N^{ABC})  \; \text{ and } \; N^{B} = \pi_{B} ( N^{ABC} ) \; \text{and} \; N^{C} = \pi_{C} ( N^{ABC} ) \, .  \]
\end{theorem}
\begin{proof}
This theorem follows immediately from Lemmas \ref{th:noumeprojasso1} and \ref{th:noumeprojasso3}.   
\end{proof}

It follows from Theorem \ref{th:assomain} that for  an arbitrary noumenal state $N^{ABC}$ of composite system $ABC$, there exist  unique noumenal states $N^{A}$ , $N^{B} $ and $N^{C}$ of respectively systems  $A$, $B$ and $C$ such that $N^{ABC}$ is equal to the product  $N^{A} \odot N^{B} \odot N^{C}$.  Of course, these states are respectively $\pi_{A} ( N^{ABC} ) $, $\pi_{B} ( N^{AB} ) $ and $\pi_{C} ( N^{ABC} ) $.

\subsection{Product of transformations}\label{sc:SEPO}
Intuitively, the next and last axiom of local-realitic theories tells us what happens to the noumenal state of  composite system $AB$, if a transformation $U$ is performed on system $A$ and a transformation $V$ is performed simultaneously on system $B$.

\begin{axiom}[Product of transformations]\label{ax:SEPO}
Associated to all disjoint systems $A$ and $B$, is  a \emph{product of transformations}, which we \mbox{denote} ``$\times_{A, B} $''. 
Given  transformations $U$ on system $A$ and $V$ on system $B$, $U \times V$, the \emph{product} of $U$ and $V$, is a transformation on system $AB$ that satisfies the following relation for any noumenal state $N^{AB} = N^{A} \odot N^{B}$ :
\[
\big( U \times_{A,B} V \big) \! \left( N^{A} \odot N^{B} \right) =  (U N^{A} )  \odot  ( V   N^{B} )  \,  .
\]
\end{axiom}
When there is no ambiguity, we drop the subscripts and write simply $\times$ instead of $\times_{A, B} $.    
Note that this equation defines $U \times V$ uniquely as the transformation satisfying the above equation because we had required the
noumenal action to be faithful; see Definition~\ref{def:faithful}. By using Theorem \ref{th:noumeproj} on the equation above,  it~follows that
\[ \pi_{A} \! \left( \left( U \times V \right) \left(  N^{A} \odot N^{B} \right) \right)= U N^{A} 
~~\text{and}~~
    \pi_{B} \! \left( \left( U \times V \right) \! \left( N^{A} \odot N^{B} \right) \right) = V N^{B} . \]
Thus, the new state of system $A$ is simply $U \, N^{A}  $, as it should.
Crucially, we see that the transformation $U$ performed on (possibly far-away) system $A$ has had abso\-lutely no effect on the noumenal state of system~$B$.

This concept is illustrated by the following commuting diagram.

\begin{center}
    \begin{tikzpicture}

      \matrix (m) [matrix of math nodes,row sep=3em,column sep=4em,minimum width=2em,ampersand replacement = \&]
      { N^{AB} \& \left( U \times V \right) \, N^{AB} \\ N^{A} \& U \, N^{A}   \\};
      \path[-stealth]
      (m-1-1) edge node [left] {$\pi_{A}$} (m-2-1)
              edge node [above] {$U \times V$} (m-1-2)
      (m-1-2) edge node [right] {$\pi_{A}$} (m-2-2)
      (m-2-1) edge node [below] {$U$} (m-2-2);

    \end{tikzpicture}

\end{center}

We now prove several properties of the product of transformations.
These proofs hinge upon the fact that if two transformations act identically on all noumenal states, then they are the same transformation,  by the  faithfulness of the noumenal action (Axiom \ref{ax:faithnoume}). Recall also that any state $N^{AB}$ can be represented as a product state $N^{AB} = N^{A} \odot N^{B}$ and that the state $N^{ABC}$ can be represented as a product $N^{ABC} =  N^{A} \odot \left( N^{B} \odot N^{C} \right) = \left( N^{A} \odot N^{B} \right) \odot N^{C}$.

\begin{theorem}\label{th:U1U2V1V2}
For any composite system $AB$, any transformations $U^{A}$ and $V^{A}$ on system $A$ and any transformations $U^{B}$ and $V^{B}$ on system $B$, 
\[ \left( U^{A} \times U^{B}  \right)  \left( V^{A} \times V^{B} \right)   =  \left( U^{A} V^{A} \right)  \times \left( U^{B} V^{B} \right)  \]
\end{theorem}
\begin{proof}
Consider arbitrary compatible noumenal states $N^{A}$ and $N^{B}$ for systems $A$ and~$B$.
\begin{align*} 
&   \left( \left( U^{A} \times U^{B} \right) \left( V^{A} \times V^{B} \right)  \right) \left( N^{A} \odot N^{B} \right) \\
=& \left( U^{A} \times U^{B} \right) \left(  \left( V^{A} \times V^{B} \right)  \left( N^{A} \odot N^{B} \right) \right) \\
= & \left( U^{A} \times U^{B} \right) \left( \left(  V^{A}   N^{A} \right) \odot \left( V^{B}   N^{B} \right) \right) \\
= & \left( U^{A} \left( V^{A}  N^{A}  \right) \right) \odot \left( U^{B} \left( V^{B}   N^{B}  \right) \right)  \\
= & \left( \left( U^{A} V^{A} \right)  N^{A} \right) \odot   \left (\left( U^{B} V^{B} \right)N^{B} \right)\\
= & \left( \left( U^{A} V^{A} \right) \times \left( U^{B} V^{B} \right) \right)  \left( N^{A}  \odot  N^{B} \right) \, . 
\end{align*}
It follows from noumenal faithfulness (Axiom \ref{ax:faithnoume}) that $\left( \left( U^{A} \times U^{B} \right) \left( V^{A} \times V^{B} \right)  \right)$ is equal to $\left( \left( U^{A} V^{A} \right) \times \left( U^{B} V^{B} \right) \right)$.
\end{proof}

\begin{theorem}\label{th:IAIBIAB}
For any composite system $AB$, 
\[ I^{A} \times I^{B}  = I^{AB}\]
\end{theorem}
\begin{proof}
Consider an arbitrary  noumenal state $N^{AB}  = N^{A} \odot N^{B}$ on disjoint systems $A$ and~$B$.
\begin{align*}
 & \left( I^{A} \times I^{B}  \right) \left( N^{A} \odot N^{B} \right)   \\
= & \, I^{A} \! \left( N^{A} \right) \odot I^{A} \! \left( N^{B} \right) \\
= & \, N^{A} \odot N^{B} \\
= & \, I^{AB} \! \left( N^{A} \odot N^{B} \right) \, . 
\end{align*}
It follows from noumenal faithfulness (Axiom \ref{ax:faithnoume}) that $I^{A} \times I^{B} $ is equal with $I^{AB}$.
\end{proof}

\begin{theorem} \label{th:commu:opprod}
For any composite system $AB$, and any transformation $U$ and $V$ on systems $A$ and $B$ respectively,
\[ U \times_{A,B} V = V \times_{B,A} U \]
\end{theorem}
\begin{proof}
Consider an arbitrary noumenal state \mbox{$N^{AB}=N^{A} \odot_{ A, B  }  N^{B}$} for system
\mbox{$AB=BA$}.
\begin{align*}
& \left( U \times_{  A,  B }  V \right) \left( N^{A} \odot_{A, B  }  N^{B} \right)  \\
= & \, ( U  N^{A} )  \odot_{ A,  B } (  V  N^{B} )  \\
= &  \, (  V  N^{B} ) \odot_{  B, A }  ( U  N^{A} ) \\
= & \left( V \times_{B, A }  U \right) \left( N^{B} \odot_{ B, A } N^{A}  \right) \\
= & \left( V \times_{ B, A }  U \right) \left( N^{A} \odot_{  A, B} N^{B}  \right) \, .
\end{align*}
It follows from noumenal faithfulness (Axiom \ref{ax:faithnoume}) that $\left( U \times_{  A,  B }  V \right) $ is equal with $\left( V \times_{    B, A }  U \right)$. 
\end{proof}

\begin{theorem}\label{th:assoc:opprod}
For any composite system $ABC$, for any transformations $U$, $V$ and $W$ on systems $A$, $B$ and $C$ respectively, 
\[ U \times \left( V \times W \right) = \left( U \times V \right) \times W  \, . \]
\end{theorem}
\begin{proof}
Consider an arbitrary noumenal state $N^{ABC} = N^{A} \odot ( N^{B} \odot N^{C}) = ( N^{A} \odot N^{B} ) \odot N^{C} $ for system $ABC$.
\begin{align*}
& \left( U \times \left(V \times W \right) \right) \left( N^{A} \odot ( N^{B} \odot N^{C} ) \right)  \\
= & \left(  U  N^{A}  \right)  \odot \left( \left( V \times W \right)  \! \left( N^{B} \odot N^{C} \right) \right) \\
= &  U  N^{A}  \odot  \left(  V  N^{B}   \odot W  N^{C}  \right) \\
= & \left( U  N^{A} \odot V N^{B}\right) \odot W  N^{C}  \\
= & \left ( \left(  U \times V \right) \! \left( N^{A} \odot N^{B} \right) \right) \odot W  N^{C} \\
= & \left( \left( U \times V \right) \times W \right) \! \left( ( N^{A} \odot N^{B} ) \odot N^{C} \right)  \, . 
\end{align*}
It follows from noumenal faithfulness (Axiom \ref{ax:faithnoume}) that $ \left( U \times \left(V \times W \right) \right) $ is equal to  $\left( \left( U \times V \right) \times W \right)$. 
\end{proof}

Since both $\odot$ and $\times$ are associative
{(Theorems~\ref{th:assoc:join} and~\ref{th:assoc:opprod})},
we can omit internal parentheses.  For example,
\[\left(  U \times V \times W \right) \left( N^{A} \odot N^{B} \odot N^{C} \right) = U  N^{A} \odot V  N^{B} \odot W  N^{B}  . \]

\paragraph{Mathematical note:}
The definition of product of transformation gives a \emph{direct product} in the usual algebraic sense;
  Had~we not required the noumenal action to be faithful,  we could have had various pathologies. For~instance, it could have happened that {even though} both $I^{AB} $ and $I^{A} \times I^{B} $ do nothing on any noumenal states, $I^{A} \times I^{B} $ is \emph{not} the neutral element of the monoid, only an element of the kernel of the noumenal action, contradicting Theorem~\ref{th:IAIBIAB}.

\subsection{No-signalling principle}\label{sc:NSP}
One important, albeit obvious, consequence of a theory being local-realistic is that it is not possible to send a signal from one system to another if there is no interaction between the two. 

Intuitively, no transformation performed on some system $A$ can have an instantaneous effect of any kind on a remote system~$B$. It~follows that no transformation performed on system $A$ can have an instantaneous \emph{observable} effect on system~$B$.
More precisely, when we perform a transformation $U$ on system $A$ and a transformation $V$ on system $B$, transformation $V$ has only affected the noumenal state of system $B$, without any influence on the noumenal state of system~$A$.  It follows that the phenomenal state of system $A$, which is a function of its noumenal state, is also unchanged.
This is formalized in the following theorem.

\begin{theorem}[No-Signalling Principle]\label{th:NSP}
Let $\rho^{AB}$ be a phenomenal state {of system~$AB$}\@. For all transformations $U$ on system $A$ and $V$ on system $B$,
\[
\pi_{A} \! \left( \left( U \times V \right) \,  \rho^{AB}  \right)  = U  \, \pi_{A} ( \rho^{AB} ) \, . 
\]
We call the equation above the \emph{no-signalling principle} because it means that no transformation $V$ applied on system $B$ can have a phenomenal (i.e.~observable) effect on a remote system~$A$. 
\end{theorem}

\begin{proof}  
Let $N^{AB}$ be any noumenal state such that $\rho^{AB}=\varphi \big( N^{AB} \big)$.
Its existence is guaranteed from the surjectivity of $\varphi$.
\begin{align*}
 & \pi_{A} \left( \left( U \times V \right) \cdot   \rho^{AB}  \right) \\
=~ & \pi_{A} \! \left( \left( U \times V \right)  \cdot  \varphi \! \left( N^{AB} \right) \right)  \\
=~ & \pi_{A} \! \left(  \varphi \left( \left( U \times V \right) \star  N^{AB}  \right) \right) \\
=~ & \varphi \left( \pi_{A} \left( \left( U \times V \right) \star  N^{AB} \right) \right) \\
=~ & \varphi \! \left( U \star  \pi_{A} \!  \left(  N^{AB}  \right) \right) \\
=~ & U  \cdot  \varphi \! \left( \pi_{A}  \! \left(N^{AB} \right) \right)  \\
=~ &  U \cdot  \pi_{A}  \! \left( \varphi  \! \left( N^{AB} \right) \right) \\ 
=~ & U \cdot \pi_{A} \!  \left( \rho^{AB} \right)  \\[-7ex]
\end{align*}
\end{proof}

Thus, a theory is no-signalling if the following diagram commutes.
\begin{center}
    \begin{tikzpicture}

      \matrix (m) [matrix of math nodes,row sep=3em,column sep=4em,minimum width=2em,ampersand replacement = \&]
      { \rho^{AB} \& \left( U \times V \right) \rho^{AB}   \\ \rho^{A} \& U  \rho^{A}   \\};
      \path[-stealth]
      (m-1-1) edge node [left] {$\pi_{A}$} (m-2-1)
              edge node [above] {$U \times V$} (m-1-2)
      (m-1-2) edge node [right] {$\pi_{A}$} (m-2-2)
      (m-2-1) edge node [below] {$U$} (m-2-2);

    \end{tikzpicture}

\end{center}

Our statement of the no-signalling principle in Theorem~\ref{th:NSP} is a generalization of the usual notion, 
which is typically formulated in terms of the probability distribution of
\emph{observation} outcomes (which would be called \emph{measurements} in quantum theory)
made in two of more remote locations.

In~the simplest bipartite instance, consider two observers Alice and Bob, who share some
system~$AB$. They dispose of sets of operations $\textsf{Operations}^{A}$ that Alice can apply
on~$A$ and $\textsf{Operations}^{B}$ that Bob can apply on~$B$.
These operations may include observations that can produce outcomes $x$ and $y$,
respectively. Denote by $\prob^{A}[ x | U]$ the probability that operation $U$
applied by Alice on system $A$ produces outcome~$x$.
Similarly, $\prob^{AB}[ x , y  | U \times V ]$ is the joint probability that
Alice observes~$x$ and Bob \mbox{observes} $y$ if they perform operations $U$ and $V$
on systems $A$ and $B$, respectively.

Assume now that Alice and Bob are sufficiently far apart that their systems can be considered
separated and non-interacting in the sense of Section~\ref{subsc:spacelike}.
The usual no-signalling principle~\cite{Gisin} says that,
for  
any possible outcome $x$ when any operation $U$ is performed by Alice on system~$A$,
\mbox{$\prob^{A}[ x  |U ]$} can be well-defined as
\[ \prob^{A}[x | U ] ~=~ \sum_{y} \prob^{AB}[ x , y | U \times V ] \, , \]
regardless of the choice of~$V$ that Bob may make.
In other words, the observable outcome at Alice's of performing some operation $U$ on system $A$
must not depend on which operation $V$ is performed by Bob on remote system~$B$,
including no operation at~all.\footnote{\,Formally, we need the identity operation to be among Bob's choices for ``including no operation at~all'' to hold.} 
It~follows that Bob cannot signal information to Alice by a clever choice of which operation to apply (or not) to his system.

\section{No-signalling theories}\label{sc:NSOT}
As previously seen, in local-realistic theories, there is a noumenal world and a phenomenal world, and these two worlds follow a parallel process.  Let us now  consider purely phenomenal theories, where there is a phenomenal world, but no underlying noumenal world. 
 More specifically, we are interested in \emph{no-signalling} theories,
in which no transformation performed on a system $A$ has any \emph{observable} effect on a disjoint system~$B$\@.
A no-signalling theory differs from a local-realistic theory in that it does not come with  noumenal state spaces. Therefore, there are no  noumenal-phenomenal epimorphisms, noumenal projectors, noumenal actions, nor a noumenal product. The latter is the essential missing ingredient in a no-signalling theory,
it has no phenomenal counterpart  which  allows  to describe the phenomenal  state of a composite system as a function of the phenomenal states of its subsystems.

We now introduce the explicit axioms that define a  no-signalling  theory.

The next five axioms appear previously as the  Axioms \ref{ax:systems}, \ref{pheno},  \ref{op}, \ref{actionpheno} and \ref{tracepheno} of a local-realistic theory.

\begin{axiom}[Systems] 
Associated to a no-signalling theory is a lattice of systems  $\left( \mathcal{S}, \sqcup, \sqcap, \overline{\,\cdot\,}, S, 0 \right)$.
\end{axiom}

\begin{axiom}[Phenomenal state space]
Associated to a system $A$  is a \emph{phenomenal state space}, $\textsf{Phenomenal-Space}^{A}$, which is a non-empty set of \emph{phenomenal states}.
\end{axiom}

\begin{axiom}[Transformations on a system]
Associated to a system  $A$   is a monoid  of transformations, $(\textsf{Transformations}^{A} , \circ^{A}, I^{A} ) $.
\end{axiom}

\begin{axiom}[Phenomenal action]
Associated to a system $A$, is a \emph{phenomenal action} ``$\cdot^{A}$'', which is an action of the monoid of transformations of the system on the set of phenomenal states of the system.
\end{axiom}

\begin{axiom}[Phenomenal projector]\label{sc:projectors}
Associated to  all systems $A$ and $B$ such that $A$ is a subsystem of $B$, is a function denoted $\pi_{A}^{B}$, which is called the \emph{ phenomenal projector} from system $B$ to system $A$.
Projector $\pi_{A}^{B}$ is a surjective function from the phenomenal space of system $B$ to the phenomenal space of system $A$.

Furthermore, for all systems $A$, $B$ and $C$, where $A$ is a subsystem of $B$, which is itself a subsystem of $C$ then the following relation hold between projectors:
\[ \pi_{A}^{B} \circ \pi_{B}^{C}= \pi_{A}^{C} , \]
where $\circ$ denotes the composition of functions.
\end{axiom}

We follow the same notational conventions introduced earlier for local-realistic theories.   For example, when there is no ambiguity, we omit  writing the phenomenal action and  we shall omit the superscript of the phenomenal projector and write simply $\pi_{A}$ instead of $\pi_{A}^{B}$.

\fussy

In local-realistic theories, the product of transformations was completely determined at the noumenal level by Axiom \ref{ax:SEPO} in Section~\ref{sc:SEPO}, which \mbox{depended} crucially on the existence of the noumenal product,
a notion that does not exist at the phenomenal level. 
Nevertheless, this induced a phenomenal meaning to the product of transformations through the noumenal-phenomenal epimorphism.

This lead to the following different axiom for the product of transformation.

\begin{axiom}[Product of transformations] \label{ax:nosigprod}
Associated to all disjoint systems $A$ and $B$, is  a \emph{product of transformations}, denoted ``$\times_{A,B}$''. However, when there is no ambiguity, we drop the subscript and write simply $\times$ to denote the product of transformation.
For all   transformations $U$ on system $A$ and $V$ on system $B$, $U \times V$, the \emph{product} of $U$ and $V$, is a transformation on system $AB$
 that satisfies the following properties:
\begin{enumerate}
\item 
\emph{No-signalling principle}.
For all  composite system $AB$, for all transformations $U$ and $V$ on  $A$ and $B$, respectively, and for all
phenomenal state $\rho^{AB}$ of ~$AB$, 
\[ \pi_{A} \! \left(  \left(  U \times V \right) \!  \rho^{AB}  \right)  = U \pi_{A} \! \left( \rho^{AB} \right)  . \]
\item
\emph{Associativity}.
For all transformations $U$, $V$ and $W$ on mutually disjoint systems, 
\[ U \times \left( V \times W \right) = \left( U \times V \right) \times W \, . \]
Since there is no ambiguity, we shall omit the parentheses and simply write $U \times V \times W$.
\item
For all composite system $AB$, and all transformations $U^{A}$, $V^{A}$ on and all transformations $U^{B}$, $V^{B}$ on $B$,
 \[ \left( V^{A} \times V^{B} \right) \left( U^{A} \times U^{B} \right) = ( V^{A} U^{A} ) \times (V^{B} U^{B} ) \, . \]
\item
For all composite system $AB$,
\[ I^{A} \times I^{B} = I^{AB} \, . \]
\item
For all composite system $AB$ and for all transformations $U^{A}$ and $U^{B}$ on  $A$ and $B$, respectively,
\[ U^{A} \times_{A,B} U^{B} = U^{B} \times_{B,A} U^{A} \, . \]
\end{enumerate}
\end{axiom}
\sloppy

These  five  properties did not have to be imposed on the product of transformations of
 local-realistic theories  because they were derived as  Theorems \ref{th:NSP}, \ref{th:assoc:opprod}, \ref{th:U1U2V1V2}, \ref{th:IAIBIAB} and~\ref{th:commu:opprod}, \mbox{respectively}.

This concludes the statement of the axioms of a no-signalling theory.

\fussy
\subsection{Postulates leading toward local-realism }

 Clearly all local-realistic theories are also no-signalling theories, since the axioms of no-signalling theories include only axioms and theorems of local-realistic theories.

We now state three additional postulates that  guarantee that any given no-signalling theory that verify them can be given a local-realistic model.  The construction of the local-realistic model for the given no-signalling theory  will be done in Section \ref{sc:main}.    We shall also state a last postulate that is not necessary to guarantee that a no-signalling theory has a local-realistic model, but that nevertheless lead to an easier construction  of the local-realistic theory.

Intuitively, the next postulate states that if nothing is done to system~$A$ and nothing to system~$B$,
then nothing is done to system~$AB$.

\begin{postulate}[Separation] \label{pos:separa}
Let $ABC$ be a composite system, let $V^{BC}$, $V^{AC}$ be transformations on systems  $BC$, $AC$ and respectively.  If  $I^{A} \times V^{BC} = I^{B} \times V^{AC} $ then there exist a transformation $V^{C}$ on system $C$ such that
\[ I^{A} \times V^{BC} =  I^{B} \times V^{AC} = I^{AB} \times V^{C} \, .  \]
\end{postulate}

While this postulate is  not derivable from the axioms of local-realistic theories,  it might be possible to find  an axiom that is natural for local-realistic theories  that implies this postulate as a theorem.

\begin{postulate}[Invertible dynamics] \label{pos:reverdy}
For every system $A$, its associated monoid of transformation $(\textsf{Transformations}^{A} , \circ^{A}, I^{A} ) $ form a group.
\end{postulate}

The previous postulate has nothing to do with local-realistic and no-signalling theories.  However, it is very natural since all modern theories of physics satisfy it.  
Importantly, this postulate says nothing about whether when a transformation   can be physically applied to a system, its inverse can also be physically applied.  What matters for the purpose of constructing  a local-realistic model  is the mathematical existence of the  inverse of a given transformation, not its physical realization. 
 To give an example, in quantum theory, the transformations associated with a system correspond to unitary operation of its Hilbert space.  It might be that when working in quantum theory a certain unitary operation can be physically applied to a system, while its inverse  cannot.  Yet, the postulate would  still apply.

The following theorem holds in an arbitrary no-signalling theory with invertible dynamics.
\begin{theorem} Let $A$ and $B$ be disjoint systems.
For any transformation $U$ on system~$A$ and $V$ on system $B$,
\[ \left( U \times V \right)^{-1} = U^{-1} \times V^{-1} \, .\]
\end{theorem}
\begin{proof}  From
\begin{align*}
  & \left( U \times V \right) \left( U^{-1} \times V^{-1} \right) \\
  & \left( U U^{-1} \right) \times \left( V V^{-1} \right) \\
  & \left( I^{A} \times I^{B} \right) \\
  & = I^{AB}  \, ,
\end{align*}
it follows from  the definition of the inverse that $(U \times V )^{-1} = U^{-1} \times V^{-1} $ .
\end{proof}

To state the last postulate, we need to define what  it means for two transformations $U$ and $V$ on a system $A$ to make absolutely no phenomenal differences.  We could be tempted to say that if for  every phenomenal state $\rho^{A}$ of $A$:
\[ U  \rho^{A}  = V \rho^{A}  \, , \]
then applying either $U$ or $V$ can make no phenomenal difference whatsoever.
However,  $U$ and $V$  could act identically on every phenomenal states $\rho^{A}$ of system $A$ and yet the transformations $U$ and $V$ could still lead to a phenomenal difference at some level. The  \href{ https://arxiv.org/abs/1710.01380v1}{first version} of this article on arXiv contained this error, which was found by Dominic Mayers, who suggested the following  definition to capture the notion that two transformations can make no phenomenal differences.

\begin{definition}[Phenomenal Equivalence] \label{def:equivpheno}
Let $A$ be a system. Le $U$ and $V$ be transformations on system $A$. 
We say that  $U$ and $V$ are \emph{phenomenally  equivalent} if
for all disjoint system $B$, for all phenomenal state $\rho^{AB}$ on composite system $AB$, for all transformation $W$ on system $B$
\[ ( U \times I^{B} ) \rho^{AB} = ( V \times I^{B}) \rho^{AB} \]
\end{definition}
It is easy to verify that phenomenal equivalence on a system is an equivalence relation on the transformations of that system.

Note that if $U$ and $V$ are phenomenally equivalent, then for every phenomenal state $\rho^{A}$ of system $A$, phenomenal state $U \rho^{A} $ is equal to phenomenal state $V \rho^{A}$ as proved by  Theorem \ref{th:eqloc} in appendix \ref{sc:faithfulness}.

The following postulate state that whenever two transformations can make absolutely no phenomenal differences, then they are the same transformation.

\begin{postulate}[Phenomenal faithfulness] \label{pos:globalfaith}
For all systems $A$,  all transformations $U$ and $V$ on system $A$  are equal whenever they  are phenomenally equivalent.
\end{postulate}

Importantly,  the phenomenal faithfulness postulate is not fundamental because any no-signalling theory  can be made phenomenally faithful by replacing transformations by equivalence classes of transformations, in effect equating any two transformations that can make no phenomenal difference whatsoever, as shown in appendix ~\ref{sc:faithfulness}.
However, it is algebraically very useful and natural to work only with no-signalling theories that satisfy phenomenal faithfulness   from the outset. 
The reason why we had not required phenomenal faithfulness in local-realistic theories is that the postulate could be incompatible with the underlying noumenal world But here, only the phenomenal world is given and we are free to build our own noumenal world to explain~it. This gives us latitude to make the phenomenal action satisfy the  phenomenally faithful postulate if needed, before we proceed to building the noumenal world, whose action will then be automatically faithful as shown by Theorem \ref{th:autofaithful}.

The following  theorem  applies in  a  theory that  satisfies all axioms of a local-realistic theory with the possible exception of the faithfulness of the noumenal action.

\begin{theorem} \label{th:nouphenoequ} 
Let $U$ and $V$ be transformations on a system $A$.   If for all noumenal states $N^{A}$ of  $A$, noumenal state $U N^{A}$ is equal to $V N^{A}$ then $U$ and $V$ are phenomenally equivalent.
\end{theorem}
\begin{proof}
Let $B$ be a system disjoint from system $A$, let $\rho^{AB}$ be an arbitrary  phenomenal state of composite system $AB$.  We must show that $ (U \times I^{B} ) \rho^{AB} =  ( V \times I^{B} ) \rho^{AB} $. By the surjectivity of the noumenal-phenomenal epimorphism, there exist  $N^{AB} = N^{A} \odot N^{B}$ a noumenal state of  $AB$ such that $\varphi (N^{AB} ) = \rho^{AB}$. 
We have that 
\begin{align*}   
( U \times I^{B} )   \rho^{AB}  & = ( U \times I^{B} ) \, \varphi ( N^{AB} )  \\
                                       &  =(  U \times I^{B} )  \,    \varphi ( N^{A} \odot N^{B} ) \\
                                    &  = \varphi (  ( U \times I^{B} )  ( N^{A} \odot N^{B} ) )  \\
                                    & = \varphi (  U  N^{A}  \odot    I^{B} N^{B} ) \\
                                    & = \varphi (  U  N^{A}  \odot  N^{B}    ) \, .  
\end{align*}
It follows that  $ (U \times I^{B} )  \rho^{AB} $ is equal to $ \varphi ( U N^{A} \odot N^{B}  )$. Similarly $( V \times I^{B} )  \rho^{AB} $ is equal to $\varphi ( VN^{A} \odot N^{B} ) $. 
Since by supposition $U N^{A} $ is equal to $V N^{A}$, it follows that 
\begin{align*}   
( U \times I^{B} )   \rho^{AB}                                      & = \varphi (  U  N^{A}  \odot  N^{B}    ) \\ 
                                                                                 & = \varphi ( V N^{A} \odot N^{B} ) \\
                                                                                & =  ( V \times I^{B} ) \rho^{AB}  \, .
\end{align*}
\end{proof}

The following  theorem  applies in  a  theory that  satisfies all axioms of a local-realistic theory with the possible exception of the noumenal faithfulness axiom.
\begin{theorem}\label{th:autofaithful}
If  phenomenal faithfulness  holds (Postulate \ref{pos:globalfaith}), then the noumenal action is  faithful on all systems.
\end{theorem}
\begin{proof}
Let $A$ be a system, let $U$ and $V$ be transformations on system $A$ such that $U  N^{A}$ is equal to  $V  N^{A}$ for all noumenal states $N$ of system $A$, we need to show that $U$ is equal to $V$.   By Theorem \ref{th:nouphenoequ}, transformations  $U$ and $V$ are phenomenally equivalent. It follows by by the phenomenal faithfulness  that $U$ is equal to $V$.
\end{proof}

The previous theorem is  the only place in the construction of a local-realistic model from a given no-signalling theory where phenomenal faithfulness is necessary.

We now introduce a last postulate which is not necessary to guarantee that a no-signalling theory can be given a local-realistic model, however, it shall lead to a simpler construction of the local-realistic model. It states that every phenomenal state of the global system is reachable from any other phenomenal state of the global system.

\begin{postulate}[Global transitivity] \label{pos:globaltrans}
For every phenomenal state $\rho$ and $\rho'$ of the global system, there exist a transformation $U$ on the global system such that \mbox{$U \rho = \rho'$}.
\end{postulate}

\subsection{Quantum theory is a no-signalling  \mbox{theory}}\label{sc:QTisNSOT}

We will soon prove that any  no-signalling theory that satisfies  postulates   \ref{pos:separa}, \ref{pos:reverdy}  and \ref{pos:globalfaith}  has a local-realistic model.   But before proving this, let us how this apply to unitary quantum theory.

First, quantum theory is a no-signalling theory. To~see this, we must define the
various components of a \mbox{no-signalling}  theory in quantum-theoretical terms.   

In unitary quantum theory,
\begin{enumerate}
\item The phenomenal state of a quantum system is its density operator.
\item The transformations acting on those states are unitary operation of the appro\-pri\-ate dimension.
\item transformation $U$ acts on phenomenal state $\rho$ by producing $U \rho \, U^{\dagger}$.
\item The phenomenal projector $\pi_A$ on system $A$ is the usual tracing out of the rest of the universe
(see Section~\ref{sc:abstract:trace}).
\item The product of transformations is the usual tensor product of unitary operations.
\end{enumerate}
For the complete details see  ref.\cite{QMlocal}.  Once the details are properly done, it is elementary to verify that all the axioms of a no-signalling theory are satisfied.  

Are the three additional postulates leading toward local-realism satisfied?

Postulate  \ref{pos:separa} holds, but its proof is slightly technical and is given in appendix  \ref{sc:postuquant}.
Invertible dynamics (Postulate \ref{pos:reverdy})  is satisfied, since the unitary operations on a system form a group.

In quantum theory, two unitary operations are phenomenally equivalent precisely whenever they  differ by a multiplicative complex number of unit norm, known as an \emph{irrelevant phase factor}. 
In~order to satisfy  phenomenal faithfulness  (Postulate \ref{pos:globalfaith}), we need to equate all unitary operations differing by a phase factor. 
The clean mathematical way to do this is given in appendix ~\ref{sc:faithfulness}.  
It follows that once this  modification is done,  quantum theory is a no-signalling theory satisfying the three postulates that guarantee the existence of a local-realistic model.  

Finally, if we restrict ourselves to the case where the phenomenal space of the global system consists of pure states only, the global transitivity (Postulate \ref{pos:globaltrans}) is also satisfied.

\section{Construction of a local-realistic model from  a no-signalling model}\label{sc:main}

We now proceed to construct  a local-realistic
model from an arbitrary given an  no-signalling model that  satisfies postulates \ref{pos:separa}, \ref{pos:reverdy} and  \ref{pos:globalfaith}.  

Specifically, we are given a lattice of systems  $\left( \mathcal{S}, \sqcup, \sqcap, \overline{\cdot} , S, 0 \right) $. 

For~each system $A$, we are given:
\begin{enumerate}
\item A phenomenal state space, $\textsf{Phenomenal-Space}^{A}$,
\item  A group of transformations  $( \textsf{Transformations}^{A}, \circ, I^{A} )$, since the invertibility postulate holds.
\item  A phenomenal action  ``$ \cdot^{A}$''.
\end{enumerate}
We~are also given a~phenomenal projector $\pi_{A}^{B}$ for each pair of systems $A$ and $B$, where $A$ is a subsystem of $B$.
Furthermore, for each pair of disjoint systems $A$ and $B$, we are given
a product of transformations $\times_{A,B}$.

Our~goal is to assign for  each system $A$:
\begin{enumerate}
\item  A noumenal state space composed of noumenal states, $\textsf{Noumenal-Space}^{A}$,
\item  A noumenal action, ``$\star^{A}$'',
\item A noumenal-phenomenal epimorphism,  $\varphi^{A}$.
\end{enumerate}

For each pair of systems $A$ and $B$, where $A$ is a subsystem of $B$, we must assign a noumenal projector $\pi_{A}^{B}$. For each pair of disjoint systems $A$ and $B$,
we  need to assign  a noumenal product $\odot_{A,B}$. Furthemore, we must verify that all axioms of a local-realistic theory are satisfied.

The proof will be done in two steps, first we will construct a model for a no-signalling theory that is globally transitive (Postulate \ref{pos:globaltrans}).  The only place in that construction where the postulate shall be assumed is in Theorem  \ref{th:episurj}.  We shall then show how to
construct a local-realistic model for a  no-signalling theory that is not necessarily globally transitive in section \ref{sc:dropping}.

\begin{definition}[Fundamental equivalence relation]
For every system $A$, we define an  equivalence relation ``\,$\sim_{A}$\,'' on  the transformations of the global system  as follows.  For any transformation $W$ and $W'$ on the global system,  $ W \sim_{A} W' $ holds whenever there exist a transformation $V$ on the complement of system  $A$  such that  
\[ W = \left( I^{A} \times V \right) W' \, . \]
\end{definition}

\begin{theorem} The relation  $\sim_{A}$ is an equivalence relation on the transformations of the global system.
 \end{theorem}
\begin{proof}
We need to show that $\sim_{A}$ is reflexive, symmetric and transitive.
 \begin{description}
 \item[The relation $\sim_{A}$ is reflexive:] Let $W$ be a transformation on the global system $S$, we show that \mbox{$W \sim_{A} W$}.  First, we see that
 \[   W = I^{S} \, W = \left(  I^{A} \times I^{ \overline{A}} \, \right) W   \, . \]
Since $W = \left( I^{A} \times I^{\overline{A}} \right) W$,
it follows that  $ W \sim_{A} W $.
\item[The relation $\sim_{A}$ is symmetric:] Let $W$ and $W'$ be transformations on the global system such that $ W \sim_{A} W' $, we show that $W' \sim_{A} W$.  By definition, there exists a transformation $V$ on the complement of system $A$  such that $ W = \left( I^{A} \times V \right) W' $.
Therefore,
\begin{align*}
W' =& ~ I^{S} W' \\
   =& \left( I^{A} \times I^{\overline{A}} \, \right) W' \\
   =& \left( I^{A}I^{A} \times V^{-1} V \right) W' \\
   =& \left( I^{A} \times V^{-1} \right) \left( I^{A} \times V \right) W' \\
   =& \left( I^{A} \times V^{-1} \right) W \, .
\end{align*}
Since $W' = \left( I^{A} \times V^{-1} \right) W$
and $V^{-1}$ is a transformation on the complement of system $A$,
it follows that  $ W' \sim_{A} W $.

\item[The relation $\sim_{A} $ is transitive:]  Let $W$, $W'$, $W''$ be transformations on the global system such that $ W \sim_{A} W' $ and $ W' \sim_{A} W'' $, we show that $W \sim_{A} W''$. By~definition, there exist transformations 
$V$ and $V'$ on the complement of system $A$ such that \mbox{$ W = \left( I^{A} \times V \right) W' $} and \mbox{$ W' = \left( I^{A} \times V' \right) W'' $}.
Therefore,
\begin{align*}
W =& \left(I^{A} \times V \right) W' \\
  =& \left( I^{A} \times V \right) \left( I^{A} \times V' \right) W'' \\
  =& \left( I^{A} \times V V' \right) W'' \, .
 \end{align*}
Since $W = \left( I^A \times V V' \right) W''$
and $V V' $ is a transformation on the complement of system $A$,
it follows that \mbox{$ W \sim_{A} W'' $}.  \\[-7ex]
\end{description}
\end{proof}

We used the invertibility postulate in the previous proof to show that the  relation $\sim_{A}$ is symmetric. The only other   places in which
the validity of our construction hinges upon that postulate is to prove that the noumenal product is well-defined in Theorem~\ref{th:welldefined} below.

Any equivalence relation gives rise to equivalence classes.

\begin{definition}[Fundamental equivalence classes] Let $A$ be a system, 
for any transformation $W $ on the global system $S$, we define the \emph{equivalence class} of $W$ with respect to $A$ to~be the set of transformations on the global system that are equivalent to $W'$ with respect to system $A$. Formally,
\[ \left[ W \right]^{A} \isdef \left\{ W' \in \textsf{Transformations}^{S} \colon W' \sim_{A} W \right\} \, . \]
\end{definition}

\begin{definition}[Noumenal states]
 The noumenal space for system $A$ is defined as 
\[ \textsf{Noumenal-Space}^{A} \isdef \left\{  \left[ W \right]^{A} \colon W \in \textsf{Transformations}^{S}  \right\}  . \]
\end{definition}
Thus, an arbitrary noumenal state $N^{A}$ is equal to $ \lbrack W \rbrack^{A}$ for some transformation $W$ on the global system $S$.

\begin{definition}[Noumenal projectors]
Let $A$ be a subsystem of $B$\@. The noumenal projector of a noumenal state $\left[ W \right]^{B}$ onto system $A$ is defined by
\[ \pi_{A}^{B}  \!\left( \left[ W \right]^{B} \right) \isdef  \left[ W \right]^{A} \, . \]
For such a definition to make sense, we need to verify that it does not depend on the choice of representative for the equivalence class. The following theorem establishes that our noumenal projectors are well defined.
\end{definition}
\begin{theorem} Let $A$ be a subsystem of $B$\@.
For any transformations  $W, W'$ on the global system, if  $ W' \sim_{B} W $ is verified, then $W' \sim_{A} W $.
\end{theorem}
\begin{proof}
Suppose $W' \sim_{B} W$.  By definition of $\sim_{B}$, there exists a transformation  $V$ on the complement of system $B$  such that $W' = ( I^{B} \times V ) W$.  Let system $C$ be equal to $ \overline{A} \sqcap B$,  systems $A$ and $C$ are disjoint, moreover composite system $AC$ is equal to $B$.  It follows that
\begin{align*} W' &= \left( I^{B} \times V \right) W \\
                        & = \left( I^{AC} \times V \right) W \\
                      & = \left(  \left( I^{A} \times I^{C} \right) \times V  \right) W \\
& = \left( I^{A} \times \left( I^{C} \times V \right) \right) W \, .
\end{align*}
Since $W' =  \left( I^{A} \times \left( I^{C} \times V \right) \right) W $, it follows that $W' \sim_{A} W$.
\end{proof}

The definition of noumenal projectors implies  that a noumenal state of system $B$ is sent to a noumenal state of subsystem $A$ by noumenal  projector $\pi_{A}^{B} $.  Thus we are justified in writing:
\[ \pi_{A}^{B} \colon \textsf{Noumenal-Space}^{B} \to \textsf{Noumenal-Space}^{A} \, . \]

When there is no ambiguity, we omit the superscript and write $ \pi_{A} $ instead of $\pi_{A}^{B}$.

It remains to verify that what we have defined satisfies the Axiom  \ref{ax:noumeproj}  for noumenal projectors.   This is a consequence of  the  next  two theorems.

\begin{theorem}\label{th:surjpartialtrace}
Noumenal projector $\pi_{A}^{B}  \colon \textsf{Noumenal-Space}^{B} \to \textsf{Noumenal-Space}^{A}$ is surjective.
\end{theorem}
\begin{proof} An arbitrary noumenal state of system $A$  is equal to  $\lbrack W \rbrack^{A} $ for some transformation $W$ on the global system $S$. Since  $ \lbrack W \rbrack^{B}$ is a noumenal state of system $B$ and $ \pi_{A}^{B} ( \lbrack W \rbrack^{B} ) = \lbrack W \rbrack^{A} $, it follows that $\pi_{A}^{B}$ is surjective.
\end{proof}

\begin{theorem}
 If $A$ is a subsystem of $B$ and $B$ is subsystem of $C$, for any arbitrary noumenal state $\lbrack W \rbrack^{C}$ of system $C$, we have 
\[  \pi_{A} \!\left( \left[ W \right]^{C} \right) = \! \left(\pi_{A} \circ \pi_{B} \right) \!\left( \left[ W \right]^{C} \right) \, . \]
\end{theorem}
\begin{proof}
 \begin{align*} \pi_{A} \!\left( \! \left[ W \right]^{C} \right) = & ~ \! \left[ W \right]^{A} \\
 = &  ~  \pi_{A} \!\left( \left[ W \right]^{B} \right) \\
 = &  ~ \pi_{A} \! \left( \pi_{B} \! \left[ W \right]^{C} \right) \\
 = & ~ \! \left( \pi_{A} \circ \pi_{B} \right) \left( \left[ W \right]^{C} \right) \\[-7ex]
 \end{align*}
\end{proof}

\begin{definition}[Noumenal action]
Let $A$ be a system and let $U$ be a transformation that acts on system $A$\@. 
We~define the {noumenal} action of transformation $U$ on system $A$ by
\[  U  \star^{A} \left[ W \right]^{A}  \isdef \left[ \left( U \times I^{\overline{A}} \, \right) W \right]^{A}  \, . \]
Again, for such a definition to make sense, we need to verify that it does not depend on the choice of representative for the equivalence class. The following theorem proves that the noumenal actions are well defined.
\end{definition}

\begin{theorem}
For any system~$A$, transformation $U$ acting on $A$, and for any  transformations $W$ and $W'$ on the global system, we have that
$W' \sim_{A} W$ implies that   $\left( U \times I^{\overline{A}} \, \right) W' \sim_{A} \left( U \times I^{\overline{A}} \, \right)  W $.
\end{theorem}
\begin{proof}
Suppose $W' \sim_{A} W$, by definition of $\sim_{A}$, there exists a transformation $V$ on the complement of system $A$ such that
\mbox{$ W' =   \left( I^{A} \times V \right)   W $}. Therefore,
\begin{align*}
 \left( U \times I^{\overline{A}} \, \right) W'  = & ~ \left( U \times I^{\overline{A}} \, \right) \left( I^{A} \times V \, \right) W \\
 = & ~ \left( U \times V \right) W \\
 = & ~ \left( I^{A} \times V \right) \left( U \times I^{\overline{A}} \, \right) W \, . \\[-7ex]
 \end{align*}
\end{proof}

The  definition  implies directly that the action of a transformation on a noumenal state is  indeed a noumenal state. It follows that for all system $A$, we are justified in writing that 
\[ \star^{A} \colon \textsf{Transformations}^{A} \times \textsf{Noumenal-Space}^{A} \to \textsf{Noumenal-Space}^{A} \, . \]
When there is no ambiguity we may omit the superscript of the action symbol and write simply $U \star \lbrack W \rbrack^{A}$ instead of $U \star^{A} \lbrack W \rbrack^{A}$. Further when there is no ambiguity, we may also omit writing the  $\star$ symbol altogether, for example writing $U \lbrack W \rbrack^{A}$ instead of $U \star \lbrack W \rbrack^{A} $ .

It remains to prove that we defined a proper action at the noumenal level and thus that Axiom \ref{axiom:noumenalaction} is  satisfied. 
This is the purpose of the following two theorems. 

\begin{theorem} For any system $A$,
for all transformations $U$ and $V$ on system $A$, and for all noumenal state $ \lbrack W \rbrack^{A}$ on system $A$, 
\[ \left( V U \right)  [ W ]^{A}  = V \! \left( U \, [ W ]^{A}  \right) \, .  \]
\end{theorem}
\begin{proof} The theorem follows from
\begin{align*}
\left( V U \right)  [ W]^{A}
 = & ~ \left[ \left( \left( VU \right) \times I^{\overline{A}} \, \right) W \right]^{A} \\
 = & ~ \left[ \left( V \times I^{\overline{A}} \, \right) \left( U \times I^{\overline{A}} \, \right) W \right]^{A} \\
 = & ~ V   \left[ \left( U \times I^{\overline{A}} \, \right) W \right]^{A}  \\
 = & ~ V \! \left( U   [ W ]^{A}   \right) \, .  \\[-7ex]
 \end{align*}
\end{proof}

\begin{theorem}
For any system $A$, for all  noumenal state $\lbrack W \rbrack^{A}$, 
\[ I^{A}  [ W ]^{A}  = [ W ]^{A}  \, . \]
\end{theorem}
\begin{proof} The theorem follows from
\[   I^{A}  [ W]^{A}  
 =  [ ( I^{A} \times I^{\overline{A}} \, ) W ]^{A} 
= \left[ I^{S} \, W \right]^{A}
 =   [ W ]^{A} \, .  \]
\end{proof}

The fidelity of the noumenal action (Axiom \ref{ax:faithnoume}) is a direct consequence of the phenomenal faithfulness of the theory (Postulate \ref{pos:globalfaith}) as proven by Theorem \ref{th:autofaithful}.

\paragraph{Noumenal product.}
We~are now ready to define the noumenal product.
First note that two noumenal states $N^{A}$ and $N^{B}$ of disjoint systems $A$ and $B$ are compatible if and only if there exists some transformation $W$ on the global system  such that $N^{A}$ is equal to  $[W]^{A}$ and $N^{B}$ is equal to~$[W]^{B}$.

\begin{definition}[Noumenal product]
Let $\left[ W \right]^{A}$ and $\left[ W \right]^{B}$ be compatible  noumenal states  of  disjoint systems $A$ and~$B$.
Their noumenal product is defined as follow
 \[ \left[ W \right]^{A} \odot_{A,B} \left[ W \right]^{B} \isdef \left[ W \right]^{AB} \, .  \]
\end{definition}

We immediately drop the subscripts in the noumenal product and write $\odot$ instead of $\odot_{A,B}$.
\noindent
Once again, for such a definition to make sense, we need to verify that it does not depend on the choice of representatives for the equivalence classes. The following theorem establishes that this is the case.

\sloppy
\begin{theorem}\label{th:welldefined}
For any transformations $W$ and $W'$ on the global system, and any composite system $AB$, 
if \mbox{$W \sim_{A} W'$} and \mbox{$W \sim_{B} W'$}, then \mbox{$W \sim_{AB} W'$}.
\end{theorem}
\begin{proof}
Let $W'$ be such that \mbox{$W \sim_{A} W'$} and \mbox{$W \sim_{B} W'$},
and let $C=\overline{AB}$. 
This means that there exist $V^{BC}$ and $V^{AC}$ such that
\mbox{$W= \! \left( I^{A} \times V^{BC} \right) W'$} and
\mbox{$W= \! \left( I^{B} \times V^{AC} \right) W'$}.
Multiplying on the right  by the inverse of $W'$, it follows that \mbox{$I^{A} \times V^{BC} = I^{B} \times V^{AC}$}.
Postulate \ref{pos:separa} imposes the
existence of a transformation $V^{C}$ such that $I^{AB} \times V^{C} = I^{A} \times V^{BC} $.  It follows that \mbox{$W = \! \left( I^{AB} \times V^{C} \right) W'$} and
therefore \mbox{$W \sim_{AB} W'$}.
\end{proof}

\fussy

The noumenal product satisfies  Axiom \ref{ax:join}  as proved  by the next theorem.

\begin{theorem}\label{th:joinworks}
Let $AB$ a composite system, and let $ \lbrack W \rbrack^{AB} $ be an arbirary noumenal state of system $AB$. It follows that
\[ \pi_{A} ( \lbrack W \rbrack^{AB} ) \odot \pi_{B}  (  \lbrack W \rbrack^{AB} )  = \lbrack W \rbrack^{AB} \, .  \]
\end{theorem}
\begin{proof} The theorem follows from
\[ \pi_{A} ( \lbrack W \rbrack^{AB} ) \odot \pi_{B} (  \lbrack W \rbrack^{AB} ) = \lbrack W \rbrack^{A} \odot \lbrack W \rbrack^{B} = \lbrack W \rbrack^{AB} \, .  \]
\end{proof}

The next theorem proves that the product of transformations  satisfies Axiom \ref{ax:SEPO}.

\begin{theorem}\label{th:main}
Let $\left[ W \right]^{A}$ and $\left[ W \right]^{B}$ be noumenal states for disjoint systems $A$ and~$B$,
and let $U$ and $V$ be transformations that can act on these  systems, respectively.  It follows that
\[ \left( U \times V \right) \!  \left( \left[ W \right]^{A} \odot \left[ W \right]^{B} \right) = U \!  \left[ W \right]^{A} \, \odot \, V \!  \left[ W \right]^{B} \, .  \]
\end{theorem}
\begin{proof}
Let $C$ be the complement of system $AB$, and thus $ABC$ is the global system. We have
\begin{align*} & \left( U \times V \right) \left( \left[ W \right]^{A} \odot \left[ W \right]^{B} \right)  \\
= & \left( U \times V \right)  \left[ W \right]^{AB}  \\
=& \left[ \left( U \times V \times I^{C} \, \right) W \right]^{AB} \\
=& \left[ \left( U \times V \times I^{C} \, \right) W \right]^{A} \, \odot \, \left[ \left( U \times V \times I^{C} \,\right) W \right]^{B} \\
=& \left[ ( U \times I^{B} \times I^{C} ) ( I^{A} \times V \times I^{C} )  W \right]^{A} \, \odot \, \left[ ( I^{A} \times V \times I^{C} )  ( U \times I^{B} \times I^{C} )  W \right]^{B} \\
=& U  \! \left[ ( I^{A} \times V \times I^{C} \, ) W \right]^{A} \, \odot \, V \! \left[ ( U \times I^{B} \times I^{C} ) W \right]^{B} \\
=& ~   U \!   \left[ W \right]^{A} \, \odot \,  V   \! \left[ W \right]^{B}  \, .  \\[-7ex] 
\end{align*}
\end{proof}

\begin{definition}[Noumenal-phenomenal homorphisms]
For each phenomenal state $\rho$ in the phenomenal space of the global system,  and for each system $A$,
we define a noumenal-phenomenal homomorphism $\phi_{\rho}^{A}$ as follows:
\[ \phi_{\rho}^{A} \! \left(  \left[ W \right]^{A} \right)  \isdef \pi_{A} \! \left( W   \rho\right) \]
for any  noumenal state $ \left[ W \right]^{A}$,  where  $\pi_{A}$ is the phenomenal projector of the underlying no-signalling theory.
\end{definition}

When there is no ambiguity,  we shall omit writing the superscript indicating the system on which a noumenal-phenomenal homorphism act.  Thus we shall write $\phi_{\rho}$ instead of $\phi_{\rho}^{A}$.

The following theorem establishes that this definition does not depend on the choice of representative for equivalence class $\left[ W \right]^{A}$.

\begin{theorem}  Let $W$ and $W'$ be transformations on the global system, if $W' \sim_{A} W$ then it follows that $ \pi_{A} \! \left( W'  \rho   \right) = \pi_{A} \! \left( W  \rho \right)$.  
 \end{theorem}
 \begin{proof}
Suppose $W' \sim_{A} W$, by definition of $\sim_{A}$,
there exist a transformation $V$ on the complement of  system $A$  such that $W' = \left( I^{A} \times V \right) W$. By applying the no-signalling principle, we see that
\begin{align*} \pi_{A} \! \left( W'  \rho  \right)  =& ~ \pi_{A} \! \left( \left( \left( I^{A} \times V \right)  W  \right) \rho \right) \\
=& ~ \pi_{A} \! \left( \left( I^{A} \times V \right) \left( W \rho \right) \right) \\ 
= & ~ I^{A} \! \left( \pi_{A} \! \left( W \rho  \right) \right)  \\
= & ~ \pi_{A} \! \left( W \rho \right)  \, . \\[-7ex]
\end{align*}
\end{proof}

The definition of the application of $\phi_{\rho}^{A}$ on an arbitrary noumenal state $\lbrack W \rbrack^{A} $ indeed gives a phenomenal state of system $A$, namely $\pi_{A} ( W \rho ) $. Therefore we may write:  
\[ \phi_{\rho}^{A} \colon \textsf{Noumenal-Space}^{A} \to \textsf{Phenomenal-Space}^{A}\, . \]

\begin{theorem}\label{th:ishomo} For every system $A$ and for every  phenomenal state $\rho$ on the global system, the function $\phi_{\rho}^{A}$ is a noumenal-phenomenal homomorphism and satisfies
\[ U \,   \phi_{\rho}^{A} \! \left( \left[ W \right]^{A} \right)  = \phi_{\rho}^{A} \! \left(  U \left[ W \right]^{A} \right) \]
for any    noumenal state $\left[ W \right]^{A}$ and  transformation $U$ on $A$.
  \end{theorem}
\begin{proof} 
\begin{align*} 
 U \, \phi_{\rho} \! \left(  \left[ W \right]^{A} \right)   =& ~ U \,  \pi_{A} \! \left( W  \rho  \right)  \\ 
=& ~ \pi_{A} \! \left(  \left( U \times I^{\overline{A}} \, \right) \left( W  \rho  \right) \right) \\ 
=& ~ \pi_{A} \! \left( \left(  \left( U \times I^{\overline{A}} \, \right)  W \right)  \rho   \right) \\ 
= & ~ \phi_{\rho} \! \left(  \left[ \left( U \times I^{\overline{A}} \, \right) W \right]^{A} \right) \\
= & ~ \phi_{\rho} \! \left(  U  \left[ W \right]^{A} \right) \, . \\[-7ex]
\end{align*}
\end{proof}

The next theorem shows that for each $\rho$ in the phenomenal space of the global system $S$, the  noumenal-phenomenal homomorphism $\phi_{\rho}$ satisfies  Axiom    \ref{ax:relnouphe} on the relation between noumenal and phenomenal projectors.
\begin{theorem}\label{th:relnouphetrace}  Let $\rho$  be a phenomenal state belonging to  the global system $S$,
\[ \pi_{A} \! \left( \phi_{\rho} \! \left(  \left[  W \right]^{B} \right) \right) =  \phi_{\rho} \! \left(  \pi_{A} \! \left( \left[ W \right]^{B} \right) \right)
\]
for any system $B$, subsystem $A$ of system $B$  and   noumenal state $\left[ W \right]^{B}$.
\end{theorem}
\begin{proof}  
\begin{align*}
  \pi_{A} \! \left( \phi_{\rho} \! \left(  \left[ W \right]^{B} \right) \right)  =& ~ \pi_{A} \! \left( \pi_{B} \left( W \rho \right) \right) \\ 
=& ~ \left( \pi_{A} \circ \pi_{B} \right) \left( W \rho \right)  \\ 
=& ~ \pi_{A} \! \left( W  \rho  \right)  \\ 
= & ~ \phi_{\rho} \! \left( [ W ]^{B} \right) \\
= & ~ \phi_{\rho} \! \left(  \pi_{A} \! \left( [ W ]^{B} \right) \right) \, .  \\[-7ex]
\end{align*}
\end{proof}

The following theorem is valid if the given no-signalling theory is globally transitive (Postulate \ref{pos:globaltrans}).
\begin{theorem}\label{th:episurj}  Let $\rho$ be a phenomental state of the global system. For each system $A$, the function  
\[\phi_{\rho}^{A} \colon \textsf{Noumenal-Space}^{A} \to \textsf{Phenomenal-Space}^{A} \] 
is a noumenal-phenomenal-epimorphism.
\end{theorem}
\begin{proof}
Let $\rho^{A}$ be an arbitrary phenomenal state of system $A$. To prove the surjectivity, we need to find a noumenal  state $N^{A}$ of system $A$ such that $ \phi_{\rho}^{A} ( N^{A} ) = \rho^{A}$. 
 
By the surjectivity of the phenomenal projector, there exist a phenomenal state $\rho^{S}$ on the global system $S$ such that $\pi_{A} ( \rho^{S} ) = \rho^{A}$.  By the globally transitive postulate, there exist a transformation $U$ on the global system such that $ \rho^{S} = U \rho$.
Consider, noumenal state $N^{A} = \lbrack U \rbrack^{A} $ of system $A$, we have that  that \[ \phi_{\rho}^{A} ( N^{A} )=  \phi_{\rho}^{A} ( \lbrack U \rbrack^{A} ) = \pi_{A} ( U \rho ) = \pi_{A} ( \rho^{S} )  = \rho^{A} \, . \]

Since $ \phi_{\rho}^{A} ( N^{A} ) = \rho^{A}$, it follows that $\phi_{\rho}^{A}$ is surjective.
\end{proof}

Thus, we may fix an arbitrary phenomenal state  $\rho$ in the phenomenal space
of the global system S and associate to each system A the noumenal-phenomenal
epimorphism $\phi_{\rho}^{A}$ . This complete the construction of a local realistic model for any given no-signalling theory if the given no-signalling theory is 
globally transitive.

\subsection{Dropping the global transitivity postulate.}\label{sc:dropping}

If the given no-signalling theory is not globally transitive, the construction of the underlying  local-realistic model might not work, and will need to be modified.  Specifically, we shall define new noumenal spaces, and therefore  new  noumenal projectors denoted~$\pi'$, a new noumenal product denoted~$\odot'$, and a new noumenal action denoted  denoted $\star'$ and a noumenal-phenomenal epimorphism, denoted $\varphi$. These  new objects will be redefined by using the  objects previously constructed.

For any given system  $A$, we  redefine its noumenal space of as follows,
\[ \textsf{New-Noumenal-Space}^{A} \]
\[ \isdef \] 
\[ \left\{ \left( \lbrack W \rbrack^{A} , \rho \right) \colon W \in \textsf{Transformations}^{S} , \rho \in \textsf{Phenomenal-Space}^{S} \right\} \, .  \]

Thus a new noumenal state on system $A$ is of the form $ ( \lbrack W \rbrack^{A} , \rho ) $, for some transformation $W$ on the global system and some phenomenal state $\rho$ on the global system.

For any system $B$, and any subsystem $A$ of $B$ we define the projection of noumenal state $ ( \lbrack W \rbrack^{B}, \rho ) $ on system $A$ as follows
\[\pi'_{A} \! \left( \lbrack W \rbrack^{B} , \rho \right) \isdef \left( \pi_{A} ( \lbrack W \rbrack^{B} ) , \rho \right)    \, . \]

For any disjoint systems $A$ and $B$, we define the noumenal product as follows
\[\left( \lbrack W \rbrack^{A} , \rho \right) \odot' \left( \lbrack W \rbrack^{B} , \rho \right) \isdef \left( \lbrack W \rbrack^{A} \odot \lbrack W \rbrack^{B} , \rho \right)  \, .
 \]
As~before, the new noumenal product $\odot'$ is only defined on compatible states:
\mbox{$\left( N^{A} , \rho \right) \odot' \left( N^{B} , \rho' \right)$} is defined under conditions that
\mbox{$ \rho= \rho'$} , noumenal states 
$N^{A}$ and $N^{B}$ are compatible  in the original noumenal spaces,
and $A$ and $B$ are disjoint systems.

For any transformation $U$ on system $A$, we define its action on noumenal state $( \lbrack W \rbrack^{A}, \rho ) $ as follows
\[ U  \star' \left( \lbrack W \rbrack^{A} , \rho \right) \isdef \left(   U \star \lbrack W \rbrack^{A} , \rho \right)  \, . \]

For any system $A$, we define the mapping of noumenal-phenomenal epimorphism $ \varphi^{A}$ on noumenal state $ ( \lbrack W \rbrack^{A}, \rho ) $ as follows
\[\varphi^{A} ( \lbrack W \rbrack^{A} , \rho ) \isdef \phi^{A}_{\rho} ( W )  \, . \]

It is easy to verify that the new noumenal spaces, new noumenal product, new actions, new noumenal projections and the new noumenal-phenomenal epimorphism are well-defined and satisfy the axioms of local-realism. Let~us prove for example that the new noumenal product behaves properly, according to Axiom ~\ref{ax:join}. The other axioms of a local-realistic theories are verified similarly.
\begin{theorem}  For every composite system $AB$, for every operation $W$ on the global system and every phemomenal state $\rho$ on the global system,
\[\pi'_{A} \! \left( \lbrack W \rbrack^{AB} , \rho \right) \odot' \pi'_{B} \! \left( \lbrack W \rbrack^{AB} , \rho \right) = \left( \lbrack W \rbrack^{AB} , \rho \right) \, . \]
\end{theorem}
\begin{proof} The theorem follows from 
\begin{align*}
\pi'_{A} \! \left( \lbrack W \rbrack^{AB}, \rho \right) \odot' \pi'_{B} \! \left( \lbrack W \rbrack^{AB}  , \rho \right) 
&=  ( \pi_{A} ( \lbrack W \rbrack^{AB} )  , \rho ) \odot'  ( \pi_{B} ( \lbrack W \rbrack^{AB} ) , \rho ) \\
&=  \left( \lbrack W \rbrack^{A} , \rho  \right) \odot' \left( \lbrack W \rbrack^{B}, \rho \right) \\
&= \left( \lbrack W \rbrack^{A} \odot \lbrack W \rbrack^{B} , \rho \right) \\
&= \left( \lbrack W \rbrack^{AB} , \rho \right) \, .  \\[-7ex]
\end{align*}
\end{proof}

Furthermore, for each system $A$, $\varphi^{A}$ is indeed a noumenal-phenomenal epimorphism,
which was the purpose of the entire exercise.
To~prove this, it suffices to show that
$\varphi_{A}$ is a  noumenal-phenomenal \emph{homo}morphism (which is obvious) and that
it is surjective,
which is the purpose of the next theorem.

\begin{theorem}
For each system $A$, \[ \varphi^{A} \colon \textsf{New-Noumenal-Space}^{A} \to \textsf{Phenomenal-Space}^{A} \] is surjective.
\end{theorem}
\begin{proof}
Consider any phenomenal state $\rho^{A}$ of the system. By the surjectivity of phenomenal projector $\pi_{A}$ there exist  a phenomenal state $\rho$ of the global system $S$, such that $\rho^{A} = \pi_{A} ( \rho^{S}) $. Thus
\[ \varphi^{A} ( \lbrack I^{S} \rbrack, \rho ) 
 =  \phi^{A}_{\rho} ( \lbrack I^{S} \rbrack ) \\
       =         \pi_{A} ( I^{S} \rho ) \\
=  \pi_{A} ( \rho )  
 =  \rho^{A}  \, .\]
Since $ \varphi^{A} ( \lbrack I^{S} \rbrack, \rho ) = \rho^{A}$, it follows that $\varphi^{A}$ is surjective. 
\end{proof}

The existence of the required noumenal-phenomenal epimorphism is established,
which concludes the construction of a local-realistic  model  corresponding to the given no-signalling theory.

\section{Conclusion}\label{sc:conc}

We have seen previously that every local-realistic theory is a no-signalling theory.  Conversely any no-signalling theory that satisfies the postulates of separation, invertible dynamics and phenomenal faithfulness has a local-realistic model. However, phenomenal faithfulness can be satisfied by identifying phenomenally equivalent transformations  and   the separation postulate  is natural and could even have been included in the axioms of local-realistic theories with only slight  lost of elegance. Thus, the key point of this paper is that every  no-signalling theory with invertible dynamics has a local-realistic model.

Can the postulate of invertible dynamics be lifted? In other words, do all no-signalling theories still have an underlying local-realistic interpretation without this extraneous postulate? 
 It~is tempting to conjecture the affirmative, since one might think that perhaps every system could carry all information of all interactions it had with all systems in the past. However, this argument is invalid because systems are not allowed in our framework to do such a thing. For example, if an invertible transformation $U$ is  applied followed by  its inverse $U^{-1}$ on a system, it has the same effect as simply doing nothing on both the noumenal and phenomenal states of the system. It is simply impossible to distinguish between these two histories that a system might have experienced.   In the end, the invertible dynamics postulate is very natural, and satisfied by all modern theories of physics and  attempting to remove the postulate would only be  useful if it could lead to a deeper insight of physical theories, else it is a mere mathematical exercise.

If a physical  theory   can be given a  local-realistic model, does it mean that the theory is in fact local-realistic?
 As  explained to the author by David Deutsch, the notion of locality is not symmetric.
A physical theory could be given a \emph{non}-local model, simply by adding extraneous noumenal invisible entities that ``talk'' to each other instantaneously across space just for the fun of~it, without having any observable effect whatsoever. It is not meaningful to claim that a physical theory is nonlocal simply because it has some nonlocal model. Otherwise, all physical theories would be nonlocal!
To~be truly considered nonlocal, a theory must have no possible local-realistic model.
It~follows that any  physical theory that \emph{can} be given a local-realistic model \emph{is} in fact  local-realistic.

The question of whether or not quantum theory is  local-realistic should not be answered merely by providing a local-realistic model for~it, as it has been done in the past. Indeed, such an answer, while mathematically valid, fails to answer the deeper question: ``But~why is quantum theory  local-realistic?''. A~metaphysical question deserves an answer based on metaphysical principles rather than by the power of mathematics alone.
So,~why is quantum theory local-realistic? The answer is that it is a no-signalling theory with invertible dynamics.

\appendix 
 
\section{Phenomenal faithfulness}\label{sc:faithfulness}

In this appendix we show that a no-signalling  theory that might not be phenomenally faithful can be transformed into a phenomenally faithful no-signalling theory  by identifying two transformations whenever they are phenomenally equivalent. To prove this, we first state a few theorems on phenomenal equivalence.

Whenever $U$ and $V$ are phenomenally equivalent transformations on system $A$, then they can make absolutely no phenomenal difference \emph{locally} on system $A$ as stated by the following theorem that is valid in an arbitrary no-signalling theory.
\begin{theorem}\label{th:eqloc}
Let $A$ be a system and let $U$ and $V$ be phenomenally equivalent transformations on system $A$, then for all phenomenal states $\rho^{A}$ on system $A$:
\[ U \rho^{A}  = V \rho^{A}  \, .\]
\end{theorem}
\begin{proof}
Let $\rho^{A} $ be a phenomenal state of system $A$. Recall that  $\pi_{A} ( \rho^{A} )= \rho^{A}$   and that $0$ is the empty system, which is disjoint from system $A$.  By using  no-signalling and the phenomenal equivalence of $U$ and $V$, we obtain
\begin{align*}
U  \; \rho^{A}  &= U \;  \pi_{A} ( \rho^{A}) \\
 &= \pi_{A} (( U \times I^{0})  \rho^{A} ) \\
&= \pi_{A} ( ( V \times I^{0})  \rho^{A}  )  \\
&= V \;   \pi_{A} ( \rho^{A})  \\
&= V \;  \rho^{A}  \, . \\[-7ex]
\end{align*}
\end{proof}

The following theorem is valid in an arbitrary no-signalling theory.
\begin{theorem}\label{th:phenoeqcomp} Let $A$ be a system, let  $U_{1}$ and  $U_{2}$ be phenomenally equivalent transformations on system $A$, and let $V_{1}$ and $V_{2} $ be phenomenally equivalent transformations on system $A$. Transformation $U_{2} \circ U_{1}$ is phenomenally equivalent with $V_{2} \circ V_{1}$.
\end{theorem}
\begin{proof}
Let $B$ be a system disjoint from system $A$.
Let $\rho^{AB}$ be an arbitrary phenomenal state on composite system $AB$. From
\begin{align*}
(( U_{2} \circ U_{1} ) \times I^{B} ) \rho^{AB} & = (( U_{2} \times I^{B} ) ( U_{1} \times I^{B} ) ) \rho^{AB} \\
 & =( U_{2} \times I^{B} ) (  ( U_{1} \times I^{B})  \rho^{AB}  ) \\
&= (U_{2} \times I^{B} ) (  ( V_{1} \times I^{B})  \rho^{AB}  ) \\
& = (V_{2} \times I^{B} ) ( ( V_{1} \times I^{B} )  \rho^{AB} )  \\
& = ( ( V_{2} \times I^{B} ) ( V_{1} \times I^{B} ) ) \rho^{AB} \\
&= ( V_{2} \circ  V_{1} ) \times I^{B} ) \rho^{AB}  \, ,
\end{align*}
it follows that $U_{2} \circ U_{1}$ is phenomenally equivalent with $V_{2} \circ V_{1}$.
\end{proof}

The following theorem is valid in an arbitrary no-signalling theory.
\begin{theorem} \label{th:prodwelldef}
Let $AB$ be a composite system,  let $U^{A}$ and $V^{A} $ be phenomenally equivalent   transformations on $A$ and let  $U^{B}$ and $V^{B}$ be phenomenally equivalent  transformations on $B$. Transformations $U^{A} \times U^{B}$ and $V^{A} \times V^{B}$ are phenomenally equivalent transformations on   $AB$.
\end{theorem}
\begin{proof}
Let $C$ be a system disjoint from composite system $AB$ and
let $\rho^{ABC} $  be   an arbitrary phenomenal state of the system $ABC$. From
\begin{align*}
((U^{A} \times U^{B} ) \times I^{C} )  \rho^{ABC} 
&= (( U^{A} \times I^{B} \times I^{C} ) ( I^{A} \times  U^{B} \times I^{C }) ) \rho^{ABC}  \\
&=     (    U^{A} \times I^{B} \times I^{C} )   (  ( I^{A} \times  U^{B} \times I^{C} )  \rho^{ABC}  )     \\
& = (U^{A} \times I^{B} \times I^{C}  )   (   ( I^{A} \times  V^{B} \times I^{C} )  \rho^{ABC}  )     \\
& = V^{A} \times I^{B} \times I^{C}  )   (   ( I^{A} \times  V^{B} \times I^{C} )  \rho^{ABC}  )     \\
& = ( ( V^{A} \times I^{B} \times I^{C} ) ( I^{A} \times V^{B} \times I^{C} ) ) \rho^{ABC}  \\
&=((V^{A} \times V^{B} ) \times I^{C} )  \rho^{ABC}  \, ,
\end{align*}
it follows that $U^{A} \times U^{B}$ is phenomenally equivalent with $V^{A} \times V^{B} $.
\end{proof}

We now show how to transform a given no-signalling theory that might not satisfy the phenomenal faithfulness postulate into a phenomenally faithful no-signalling theory by identifying transformations whenever they are phenomenally equivalent. This modified no-signalling theory shall consist of the same lattice of system, the same phenomenal states, the  same phenomenal projections as the given theory. However the modified transformations shall be  defined as  equivalence classes of   transformations from the given theory.  This implies that the  product of transformation and the composition of transformations must be defined on equivalence classes of transformations.

Thus, we are given a lattice of systems  $\left( \mathcal{S}, \sqcup, \sqcap, \overline{\cdot} , S, 0 \right) $, where for~each system $A$, we are given:
\begin{enumerate}
\item A phenomenal state space, $\textsf{Phenomenal-Space}^{A}$,
\item  A monoid of transformations  $( \textsf{Transformations}^{A}, \circ^{A}, I^{A} )$, 
\item  A phenomenal action  ``$ \cdot^{A}$''.
\end{enumerate}
We~are also given a~phenomenal projector $\pi_{A}^{B}$ for each pair of systems where $A$ is a subsystem of $B$.
Furthermore, for each pair of disjoint systems $A$ and $B$, we are given
a product of transformations $\times_{A,B}$.

Let $U$ be a transformation on system $A$, we define $ \lbrack U \rbrack^{A} $, the \emph{phenomenal equivalence class} of $U$, to be the set of all transformations $V$ on system $A$ that are phenomenally equivalent to $U$.
Formally,
\[ \lbrack U \rbrack^{A} \isdef \{ V  \in \textsf{Transformations}^{A} \colon V \text{ is phenomenally equivalent to } U \} \, . \]
  Whenever there is no ambiguity, we drop the superscript and write $ \lbrack U \rbrack$ instead of $ \lbrack U \rbrack^{A}$.

The transformations on a  system $A$ in the modified  theory  are defined as its set of phenomenal equivalence classes. Formally,
\[ \textsf{New-Transformations}^{A} \isdef \left\{ \left[ U \right] \colon U \in \textsf{Transformations}^{A} \right\} \, . \]

Let  $A$ be a system, we define the composition of equivalence classes of transformations   by:
\[ \left[ U \right] \circ' \left[  V \right] \isdef \left[ U \circ V \right] \, , \]
for any transformations $U$ and $V$ on system $A$.
The newly defined composition is well defined since it does not depend on choice of representative as proven by the Theorem \ref{th:phenoeqcomp}.
It  is easy to verify that that  $ ( \textsf{New-Transformations} , \circ' , \lbrack I^{A} \rbrack )$ is a monoid of transformation.

Let  $A$ be a system, we define the  action of an equivalence class of transformations on phenomenal state by:
\[ \left[ U \right] \star' \rho^{A} \isdef  U \star \rho^{A} \, , \]
for any transformation $U$ on $A$ and phenomenal state $\rho^{A}$ on $A$.
The newly defined phenomenal  action is well defined since it does not depend on choice of representative as proved by Theorem \ref{th:eqloc}.  It is also easy to verify that we have indeed defined a phenomenal action and that the phenomenal faithfulness postulate is verified, which is the point of this construction.

Let $A$ and $B$ be disjoint systems, we define the product of equivalence classes of transformations by 
\[ \left[ U \right] \times' \left[ V \right] \isdef \left[ U \times V \right] \, ,  \]
for any transformation $U$ on system $A$ and any transformation $V$ on system $B$.
The newly defined  product  is well defined since it does not depend on choice of representative as proven by Theorem  \ref{th:prodwelldef}.   It is also  easy  to verify that the product of transformation verifies the five properties of  Axiom \ref{ax:nosigprod}. As an example, the next theorem shows that the no-signalling principle is satisfied.

\begin{theorem}
Let $AB$ be a composite system, 
\[ \pi_{A} ( ( \left[ U \right] \times' \left[ V \right] )  \star'  \rho^{AB}  ) = \left[ U \right] \star'  \pi_{A} ( \rho^{AB} ) \, . \]
for any pair of transformations  $U$ and $V$  on systems $A$ and $B$ respectively, and any phenomenal state $\rho^{AB}$ on $AB$.
\end{theorem}
\begin{proof}  The theorem follows from
\begin{align*}
      \pi_{A} ( ( \left[ U \right] \times' \left[ V \right]) \star'  \rho^{AB}  ) 
&=\pi_{A} ( (\left[ U \times V \right] )  \star'  \rho^{AB}  ) \\
&=  \pi_{A}(  \left( U \times V \right) \star  \rho^{AB} ) \\
&= U \star  \pi_{A}  ( \rho^{AB}  ) \\
& = \left[ U \right] \star'  \pi_{A} ( \rho^{AB} )   \, .  \\[-7ex]
\end{align*}
\end{proof}

Further, it is easy to verify that  whenever the given no-signalling theory  satisfies any postulate among invertible dynamics, separation and global transitivity then the newly constructed theory will also satisfies it.

This completes our construction, from which it follows that an arbitrary no-signalling theory can be modified into a phenomenally faithful no-signalling theory  by identifying phenomenally equivalent transformations.

\section{Noumenal faithfulness}
In this appendix we show that a theory that satisfies all axioms of local-realism with the possible exception of the faithfulness of the noumenal action (Axiom \ref{ax:faithnoume}) can be transformed into a  local-realistic theory with a faithful noumenal action  by identifying two transformations whenever they make no difference whatsover at the noumenal level. To prove this, we  first define some concepts and prove some theorems that apply in an arbitrary theory that satisfies all axioms of local-realism with the possible exception of the faithfulness of the noumenal action.

\begin{definition} Let $A$ be a system, let $U$ and $V$ be transformations on system $A$, we say that $U$ and $V$ are \emph{noumenally equivalent} on system $A$, whenever for all noumenal state $N^{A}$ on system $A$, noumenal state $ U N^{A}$ is equal to $V N^{A}$.
\end{definition}
It is easy to verify that noumenal equivalence on a system is an equivalence relation on the transformations of that system.  Note that if transformations $U$ and $V$ are noumenally equivalent, they are automatically phenomenally equivalent by Theorem \ref{th:nouphenoequ}.

\begin{theorem} \label{th:phacti}
Let $A$ be a system, let $U$ and $V$ be noumenally equivalent transformation on $A$, then for all phenomenal states $\rho^{A}$ on system $A$:
\[ U \rho^{A} = V \rho^{A} \, . \]
\end{theorem}
\begin{proof}   Since transformations $U$ and $V$ are also phenomenally equivalent,  the result follows directly by applying Theorem \ref{th:eqloc}.
\end{proof}

\begin{theorem} \label{th:comnoume}
Let $A$ be a system,  let $U_{1}$ and $V_{1}$ are noumenally equivalent transformation on system $A$   and  $U_{2}$ and $V_{2}$ are noumenally equivalent transformations on system $A$.  Transformation $U_{2} \circ U_{1}$ is noumenally equivalent to $ V_{2} \circ V_{1} $.
\end{theorem} 
\begin{proof}
Let $N^{A}$ be an arbitrary noumenal state of system $A$.  Since
\begin{align*}
(U_{2} \circ U_{1}) N^{A}  &= U_{2} ( U_{1} N^{A} )\\
&= U_{2} ( V_{1} N^{A} )   \\
&= V_{2} (V_{1}  N^{A}) \\
&= ( V_{2} \circ V_{1}) N^{A} \, ,
\end{align*}
it follows that $U_{2} \circ U_{1}$ is noumenally equivalent to $ V_{2} \circ V_{1}$ .
\end{proof}

\begin{theorem} \label{th:wdprod}
Let $AB$ be a composite system,  let $U^{A}$ and $V^{A} $ be noumenally equivalent   transformations on $A$ and let  $U^{B}$ and $V^{B}$ be  noumenally equivalent  transformations on $B$. It follows that transformations $U^{A} \times U^{B}$ and $V^{A} \times V^{B}$ are noumenally equivalent transformations on   $AB$.
\end{theorem}
\begin{proof}
Let $N^{AB} = N^{A} \odot N^{B} $ be an arbitrary noumenal state on system $AB$.  Since
\begin{align*} 
(U^{A} \times U^{B} ) ( N^{A} \odot N^{B} )& = ( U^{A}  N^{A} ) \odot (U^{B}  N^{B} ) \\
&= ( V^{A} N^{A} ) \odot ( V^{B} N^{B} ) \\
&=  ( V^{A} \times V^{B} ) ( N^{A} \odot N^{B} ) \, , \end{align*}
it follows that $U^{A} \times U^{B}$ is noumenally equivalent to $V^{A} \times V^{B}$.
\end{proof}

 We now show how to transform a  a theory that satisfies all axioms of local-realistic theory with the possible exception of the faithfulness of the noumenal action  (Axiom \ref{ax:faithnoume}) into a local-realistic theory where the noumenal action \emph{is}  faithful by identifying transformations whenever they are noumenally equivalent. This modified local-realistic theory shall consist of the same lattice of system, the same noumenal states, the same phenomenal states, the  same projections as the given theory, the same noumenal product and the same noumenal-phenomenal epimorphism. However the modified transformations shall be  defined as  equivalence classes of   transformations from the given theory.  This implies that the  product of transformations and the composition of transformations will be defined on equivalence classes of transformations.

Thus, we are given a lattice of systems  $\left( \mathcal{S}, \sqcup, \sqcap, \overline{\cdot} , S, 0 \right) $, where for~each system $A$, we are given:
\begin{enumerate}
\item A noumenal state space, $\textsf{Noumenal-Space}^{A}$,
\item A phenomenal state space, $\textsf{Phenomenal-Space}^{A}$,
\item  A monoid of transformations  $( \textsf{Transformations}^{A}, \circ^{A}, I^{A} )$, 
\item A noumenal action ``$ \star^{A}$'',
\item  A phenomenal action  ``$ \cdot^{A}$'',
\item A noumenal-phenomenal epimorphism $\varphi^{A}$.
\end{enumerate}
We~are also given a~ projector $\pi_{A}^{B}$ for each pair of systems where $A$ is a subsystem of $B$.
Furthermore, for each pair of disjoint systems $A$ and $B$, we are given
a product of transformations $\times_{A,B}$ and a noumenal product $\odot_{A,B} $.

Let $U$ be a transformation on system $A$, we define $ \lbrack U \rbrack^{A} $, the \emph{noumenal equivalence class} of $U$, to be the set of all transformations $V$ on system $A$ that are noumenally equivalent to $U$.
Formally,
\[ \lbrack U \rbrack^{A} \isdef \{ V  \in \textsf{Transformations}^{A} \colon V \text{ is noumenally equivalent to } U \} \, . \]
  Whenever there is no ambiguity, we drop the superscript and write $ \lbrack U \rbrack$ instead of $ \lbrack U \rbrack^{A}$.

The transformations on a  system $A$ in the modified  theory  are defined as its set of noumenal equivalence classes. Formally,
\[ \textsf{New-Transformations}^{A} \isdef \left\{ \left[ U \right] \colon U \in \textsf{Transformations}^{A} \right\} \, . \]

Let  $A$ be a system, we define the composition of equivalence classes of transformations   by:
\[ \left[ U \right] \circ' \left[  V \right] \isdef \left[ U \circ V \right] \, , \]
for any transformations $U$ and $V$ on system $A$.
The newly defined composition is well defined since it does not depend on choice of representative as proven by the Theorem \ref{th:comnoume}.
It  is easy to verify that that  $ ( \textsf{New-Transformations} , \circ' , \lbrack I^{A} \rbrack )$ is a monoid of transformations.

Let  $A$ be a system, we define the  action of an equivalence class of transformations on noumenal states by:
\[ \left[ U \right] \star' N^{A} \isdef  U \star N^{A} \, , \]
for any transformation $U$ on $A$ and phenomenal state $\rho^{A}$ on $A$.
The newly defined noumenal action is well defined since it does not depend on choice of representative.  It is also easy to verify that we have indeed defined a faithful noumenal action, which is the point of this construction.

Let  $A$ be a system, we define the  action of an equivalence class of transformations on phenomenal states by:
\[ \left[ U \right] \cdot' \rho^{A} \isdef  U \cdot \rho^{A} \, , \]
for any transformation $U$ on $A$ and phenomenal state $\rho^{A}$ on $A$.
The newly defined phenomenal  action is well defined since it does not depend on choice of representative as proved by  Theorem \ref{th:phacti}.  It is also easy to verify that we have indeed defined a phenomenal action.

Let $A$ and $B$ be disjoint systems, we define the product of equivalence classes of transformations by 
\[ \left[ U \right] \times' \left[ V \right] \isdef \left[ U \times V \right] \, ,  \]
for any transformation $U$ on system $A$ and any transformation $V$ on system $B$.
The newly defined  product  is well defined since it does not depend on choice of representative as proven by Theorem  \ref{th:wdprod}.   This product of transformation satisfy Axiom \ref{ax:SEPO} as proved by the next theorem.

\begin{theorem}Let $AB$ be a composite system, let $U$ and $V$ be transformations on  $A$ and $B$ respectively,  let $N^{AB} = N^{A} \odot N^{B}$ be  an arbitrary noumenal state on $AB$,
\[  (  \lbrack U \rbrack  \times'  \lbrack V \rbrack ) \star' (N^{A} \odot N^{B} ) = ( \lbrack U \rbrack \star' N^{A} ) \odot ( \lbrack V \rbrack \star' N^{B} ) \] 
\end{theorem}
\begin{proof} The theorem follows from
\begin{align*}
 (  \lbrack U \rbrack  \times'  \lbrack V \rbrack ) \star' (N^{A} \odot N^{B} ) & = (  \lbrack U  \times   V \rbrack ) \star' ( N^{A} \odot N^{B} ) \\
                                                                                                              & = ( U \times  V) \star (N^{A} \odot N^{B}) \\
                                                                                                                & = ( U \star N^{A} ) \odot ( V \star N^{B} ) \\
                                                                                                                & =  ( \lbrack U \rbrack \star' N^{A} ) \odot ( \lbrack V \rbrack \star' N^{B} )  \, . \\[-7ex]
\end{align*}
\end{proof}

It is similarly easy to verify that all axioms of a local-realistic theory are verified.
This completes our construction, from which it follows that the given theory  can be modified into a local-realistic theory with a faithful noumenal action by identifying noumenally equivalent transformations.

\section{Additional postulates of no-signalling theories in Quantum theory} \label{sc:postuquant}

In this appendix we prove that  separation (Postulate \ref{pos:separa}) holds in quantum theory.
Before doing so, we first recall some facts on vector spaces.

Given a vector space $V$ over a field $\mathbb{F}$, we shall denote by  $\mathcal{L} ( V) $  the set of linear maps from $V$ to to itself. 
 More formally,
\[ \mathcal{L} ( V )  = \{ L  \;  | \; L \colon V \to V \text{ is a linear map} \} \,  . \]

It is well known that  $\mathcal{L} ( V ) $ is itself a vector space over the field $\mathbb{F}$.

\begin{definition}[Product of  bases]  Let $V_{1} $ and $V_{2}$ be vector spaces over a field $\mathbb{F} $.
Let $\mathcal{B}_{1}$ be a basis of  $\mathcal{L} ( V_{1})$ and $\mathcal{B}_{2} $ be a basis  of  $\mathcal{L} ( V_{2} )  $. We define the tensor product of $\mathcal{B}_{1} $ and $\mathcal{B}_{2} $ to be
\[ \mathcal{B}_{1} \otimes \mathcal{B}_{2} \stackrel{\text{def}}{=} \{ L_{1} \otimes L_{2} \colon L_{1} \in \mathcal{B}_{1} , L_{2} \in \mathcal{B}_{2}  \} \, . \] 
\end{definition}
It is well-known  that $\mathcal{B}_{1} \otimes \mathcal{B}_{2}$ is a  basis of  $\mathcal{L} ( V_{1} \otimes V_{2}) $.

The next theorem implies the separation postulate in quantum theory and its proof is due to Michel Boyer  \cite{BOYER}.

\begin{theorem}
Let $\mathcal{H}^{A}$, $\mathcal{H}^{B}$ and $\mathcal{H}^{C}$ be Hilbert spaces, let $V^{AB}$ and $V^{BC}$ be unitary operations on  $\mathcal{H}^{A} \otimes \mathcal{H}^{B} $ and $\mathcal{H}^{B}\otimes \mathcal{H}^{C}$ respectively.  If
\[ V^{AB} \otimes I^{C}  =   I^{A} \otimes V^{BC} \, ,  \]
then there exist $V^{B}$, a unitary operation on $\mathcal{H}^{B}$ such that 
\[  V^{AB} \otimes I^{C}=I^{A} \otimes V^{B} \otimes I^{C}    =   I^{A} \otimes V^{BC} \, .   \]
\end{theorem}
\begin{proof}
Let $\mathcal{B}^{A}$  be a basis of  $\mathcal{L} ( \mathcal{H}^{A} ) $, such that $I^{A}$ is an element of $\mathcal{B}^{A}$, let $\mathcal{B}^{C}$ be  a basis of  $\mathcal{L} ( \mathcal{H}^{C} )$ such that $I^{C}$ is an element of $\mathcal{B}^{C}$ and  let $\mathcal{B}^{B}$ be an arbitrary basis of $\mathcal{L} ( \mathcal{H}^{B} )$.   Thus  $\mathcal{B}^{A} \otimes \mathcal{B}^{B} \otimes \mathcal{B}^{C}$ is a basis of $\mathcal{L} ( \mathcal{H}^{A} \otimes \mathcal{H}^{B} \otimes \mathcal{H}^{C}  ) $, therefore  $ V^{AB} \otimes I^{C} = I^{A} \otimes V^{BC}$ can be uniquely represented as a sum of  elements of it:
\[ V^{AB} \otimes I^{C} = I^{A} \otimes V^{BC}  = \sum\limits_{ijk} c_{ijk} B^{A}_{i} \otimes B^{B}_{j} \otimes B^{C}_{k}  \, . \]

Since \[ I^{A} \otimes V^{BC} =  \sum\limits_{ijk} c_{ijk} B^{A}_{i} \otimes B^{B}_{j} \otimes B^{C}_{k}  \, , \] 
this  implies that for all $i$ such that $B^{A}_{i} \neq I^{A}$ , the coefficient $c_{ijk}$ is equal to zero.

Equally since \[ V^{AB} \otimes I^{C} =  \sum\limits_{ijk} c_{ijk} B^{A}_{i} \otimes B^{B}_{j} \otimes B^{C}_{k}  \, , \] 
this implies that for all $k$ such that $B^{C}_{k} \neq I^{C}$ , the coefficient $c_{ijk}$ is equal to zero.

Let $i_{0}$ such that $I^{A} = B^{A}_{i_{0}}$ and let $k_{0}$ such that $I^{C} = B^{C}_{k_{0}}$.
By removing null coefficients, we obtain
 \[ V^{AB} \otimes I^{C} =  I^{A} \otimes V^{BC} =  \sum\limits_{j} c_{ i_{0} j k_{0}} I^{A} \otimes B^{B}_{j} \otimes I^{C} \, . \] 

Let  $V^{B}$ be equal to $ \sum\limits_{j} c_{i_{0}j k_{0}}  B^{B}_{j} $. Thus, the equation  above becomes
 \[ V^{AB} \otimes I^{C} = I^{A} \otimes V^{BC} = I^{A} \otimes V^{B} \otimes I^{C} \, .   \]
 Because $V^{BC} $ is unitary, $V^{C} $ must also be unitary.
\end{proof}

\end{document}